\documentclass[sigconf]{acmart}
\AtBeginDocument{%
  }

\setcopyright{acmlicensed}
\copyrightyear{2018}
\acmYear{2018}
\acmDOI{XXXXXXX.XXXXXXX}
\acmConference[Conference acronym 'XX]{Make sure to enter the correct
  conference title from your rights confirmation email}{June 03--05,
  2018}{Woodstock, NY}
\acmISBN{978-1-4503-XXXX-X/2018/06}

\usepackage{enumitem}
\usepackage{booktabs}
\usepackage[normalem]{ulem}
\usepackage{multirow}
\definecolor{myblue}{HTML}{268BD2}
\definecolor{mygreen}{HTML}{658354}
\definecolor{results}{RGB}{220, 230, 240}
\usepackage[table]{xcolor}
\usepackage{tcolorbox}
\usepackage{algorithm}
\usepackage{algorithmic}

\usepackage{color, xspace}

\newcommand{\name}{GloRank}

\begin{document}

\title{From Local Indices to Global Identifiers: Generative Reranking for Recommender Systems via Global Action Space}

\author{Pengyue Jia$^1$*, Xiaobei Wang$^2$*, Yingyi Zhang$^1$*, Shuchang Liu$^2$, Yupeng Hou$^3$, Hailan Yang$^2$}
\author{Xu Gao$^2$, Xiaopeng Li$^1$, Yejing Wang$^1$, Julian McAuley$^3$, Xiang Li$^2$, Lantao Hu$^2$}
\author{Yongqi Liu$^2$, Kaiqiao Zhan$^2$, Han Li$^2$, Kun Gai$^2$, Xiangyu Zhao$^1$}

\affiliation{
	\institution{$^1$City University of Hong Kong, $^2$Kuaishou Technology, $^3$University of California San Diego}
	\country{}
}

\email{jia.pengyue@my.cityu.edu.hk, liushuchang@kuaishou.com, xianzhao@cityu.edu.hk}

\thanks{$*$ Both authors contributed equally to this work.}

\renewcommand{\shortauthors}{Jia et al.}

\begin{abstract}
In modern recommender systems, list-wise reranking serves as a critical phase within the multi-stage pipeline, finalizing the exposed item sequence and directly impacting user satisfaction by modeling complex intra-list item dependencies.
Existing methods typically formulate this task as selecting indices from the local input list.
However, this approach suffers from a semantically inconsistent action space: the same output neuron (logits) represents different items across different samples, preventing the model from establishing a stable, intrinsic understanding of the items.
To address this, we propose \name~(\textbf{Glo}bal Action Space \textbf{Rank}er), \textit{a generative framework that shifts reranking from selecting local indices to generating global identifiers}.
Specifically, we represent items as sequences of discrete tokens and reformulate reranking as a token generation task.
This design effectively decouples the scoring mechanism from the variable input order, ensuring that items are evaluated against a consistent global standard.
We further enhance this with a two-stage optimization pipeline: a supervised pre-training phase to initialize the model with high-quality demonstrations, followed by a reinforcement learning-based post-training phase to directly maximize list-wise utility.
Extensive experiments on two public benchmarks and a large-scale industrial dataset, coupled with online A/B tests, demonstrate that \name~consistently outperforms state-of-the-art baselines and achieves superior robustness in cold-start scenarios.
\end{abstract}



\maketitle

\section{Introduction}

In modern recommender systems, list-wise reranking~\cite{pei2019prm,liu2022rerankingsurvey} is typically applied as the final step of a multi-stage recommendation pipeline~\cite{covington2016youtube,gallagher2019cascade,liu2025recflow}.
While the retrieval and intermediate ranking phases primarily aim to filter the candidate pool, the reranking stage is dedicated to optimizing the collective utility of the list by modeling intra-list item dependencies.
Therefore, the reranking stage can be formulated as a combination optimization problem~\cite{bello2018seq2slate,zhang2025goalrank}, which selects and orders $K$ items from a candidate set of size $N$, creating a permutation space of size $N!/(N-K)!$, where $N$ typically ranges from dozens to hundreds in industrial settings.
In such cases, this factorial complexity renders the search for an optimal item combination computationally intractable, particularly under strict latency constraints. 
Moreover, user feedback is observed only for the single displayed list, which makes supervision extremely sparse with respect to the underlying permutation space~\cite{liu2022rerankingsurvey}. 
These characteristics make it difficult to design reranking models that are both effective and stable during training.

To address these challenges, existing approaches first studied generator-only methods~\cite{bello2018seq2slate,pei2019prm,lin2024dcdr} that focus on directly producing a ranked list, typically by employing autoregressive models or generative decoding strategies.
These approaches are characterized by their end-to-end simplicity and direct generation capability.
More recently, based on the intuition that discriminating high-quality lists is easier than directly generating them, the generator-evaluator paradigm has gained prominence in industrial systems~\cite{shi2023pier,ren2024nar4rec,yang2025mgeclig}.
In this framework, the process is decoupled into two phases: a generator or multiple generators first propose a set of candidate list permutations, and a separate evaluator subsequently selects the optimal sequence based on a list-wise value function, thereby explicitly modeling inter-item interactions in the generated list.

\begin{figure}[t]
    \centering
    \includegraphics[width=\linewidth]{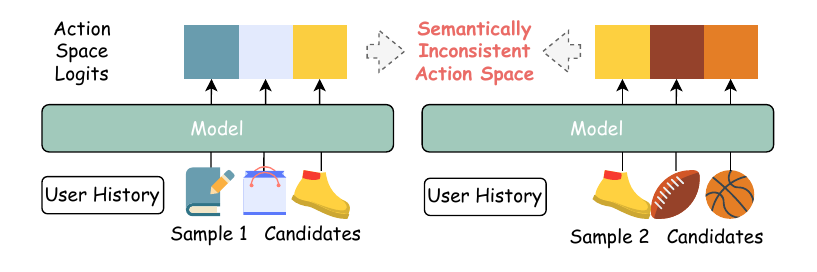}
    \caption{Semantic Inconsistency in Action Spaces.}
    \label{fig:intro}
\end{figure}

Despite these advances, existing architectures~\cite{zhang2025goalrank,yang2025mgeclig} typically operate within a \textbf{position-dependent action space}. 
In this formulation, the model learns to select items based on their relative indices (e.g., the $k$-th candidate) rather than their intrinsic identities.
This creates a semantically inconsistent mapping: the same output neuron (logit) corresponds to completely different items across different samples, depending on the randomized input order, as shown in Figure~\ref{fig:intro}. 
Consequently, the model struggles to capture the intrinsic utility of specific items (e.g., ``the iPhone is desirable''). Instead, it tends to overfit spurious patterns associated with the indices themselves (e.g., memorizing ``the $3^{\text{rd}}$ position is statistically preferable''), leading to poor generalization.

To overcome this, we argue that robust reranking requires a paradigm shift from \textit{selecting local indices} to \textit{generating global identifiers}. 
We need an architecture where the action of selecting an item is defined by its fixed semantic identity, independent of its temporary position in the candidate list.
However, implementing this global identifier generation paradigm is non-trivial: industrial corpora often scale to millions of items, rendering a flat global vocabulary computationally prohibitive for standard generative models. 
This raises a fundamental challenge: \textit{How can we efficiently anchor reranking in a consistent global space while bypassing the scalability bottlenecks of a massive vocabulary?}

To answer this question, we propose \name~(\textbf{Glo}bal Action Space \textbf{Rank}er), a novel generative framework redefines reranking by moving from local index selection to global identifier generation. 
Specifically, we employ Semantic IDs (SIDs)~\cite{fu2025forge} to represent items as sequences of discrete tokens derived from item attributes. 
This transformation maps the million-scale item corpus into a compact, fixed global vocabulary. 
Consequently, the reranking task is reformulated from \textit{selecting dynamic local indices} to \textit{generating static global identifiers}. 
This not only ensures that the model learns a stable evaluation standard that directly correlates with item semantics, but also is beneficial for improving cross-stage modeling consistency since SIDs are widely adopted in retrieval stages in many modern designs~\cite{rajput2023recommender,deng2025onerec}.
To ensure training stability and maximize performance, we introduce a two-stage optimization pipeline (inspired by the modern training framework of large recommendation models~\cite{zhou2025onerecv2,zhang2025goalrank}) adapted to the reranking task. 
A supervised pre-training phase first initializes the model by imitating high-quality reference lists, and a reinforcement learning-based post-training phase then refines the model to directly optimize list-wise rewards.
In addition, we further design a constrained decoding scheme, which maintains generation trees over the candidate sets, ensuring valid and non-duplicated output items.
We conduct both offline and online experiments to verify the effectiveness of \name~compared to state-of-the-art baselines.
In summary, the main contributions of this work are as follows: 
\begin{itemize}[leftmargin=*]
    \item We identify the semantically inconsistent action space as a fundamental bottleneck in existing reranking paradigms. We analyze how their reliance on dynamic relative indices causes the semantic meaning of output nodes to vary with input order, preventing the model from establishing a stable valuation standard for items.
    
    \item We propose \name~, a simple and effective generative reranking framework that advocates for a paradigm shift from local index selection to global identifier generation, by modeling the output space in a fixed global vocabulary, training with a two-stage paradigm, and inference with constrained decoding. 
    
    \item Extensive experiments on two public benchmarks and a large-scale industrial dataset, coupled with online A/B tests, demonstrate that \name~consistently outperforms state-of-the-art baselines in standard evaluation and achieves superior robustness in cold-start scenarios.
\end{itemize}

\section{Theoretical Analysis}
\label{sec:theory}

This section analyzes an optimization instability in list-wise reranking that originates from the semantic definition of the action space at the output layer.
We show that, under (candidate) position-dependent action spaces, the supervision signal received by a fixed output parameter row becomes stochastic even when the target item identity is fixed.
This induces an additional, irreducible variance term in the output-layer gradients, resulting in training instability.
In contrast, anchoring supervision in a globally consistent item space removes this mapping-induced variance, leaving only the variance caused by the encoder representation itself.

\subsection{List-wise Reranking Setup}

Let $\mathcal{V}$ denote the global item universe.
For each request with context $x$, the reranking stage receives a candidate set
$\mathcal{C}=\{v_1,\dots,v_N\}\subset\mathcal{V}$ and outputs an ordered list
$R=[r_1,\dots,r_K]$, where $K\le N$.
We consider an autoregressive list-wise generator with step-level distribution
$
p_\theta(r_t \mid x,\mathcal{C}, r_{<t}),
$
trained by minimizing the sum of step losses (e.g., cross-entropy) over a target list $R^*$:
\begin{equation}
\label{eq:listwise_ce}
\mathcal{L}(\theta)
=
\sum_{t=1}^{K}
\ell\!\left(\mathbf{z}^{(t)}_\theta,y^{(t)}\right),
\end{equation}
where $\mathbf{z}^{(t)}_\theta \in \mathbb{R}^{A}$ denotes the output logits of $p_\theta$ and $A$ denotes the size of the action space.
Typically, the output logits are calculated from a hidden state:
\begin{equation}
\label{eq:logits}
\mathbf{z}^{(t)}_\theta(\sigma) = \mathbf{W}\,\mathbf{h}^{(t)}_\sigma,
\end{equation}
where $\mathbf{h}^{(t)}_\sigma=h^{(t)}_\phi(x,\sigma(\mathcal{C}),r_{<t})\in\mathbb{R}^d$ is the hidden state before output,
$\mathbf{W}\in\mathbb{R}^{A\times d}$ is the output-layer parameter matrix,
and $\sigma$ denotes a permutation of the candidate list used to form the ordered input $\sigma(\mathcal{C})$.
We explicitly allow $\mathbf{h}^{(t)}_\sigma$ to depend on $\sigma$ (for example, due to position embeddings and attention over ordered inputs).

For cross-entropy, the gradient with respect to $\mathbf{W}$ is
\begin{equation}
\label{eq:grad_W}
\nabla_\mathbf{W}\,\ell\big(\mathbf{z}^{(t)}(\sigma), y^{(t)}(\sigma)\big)
=
\left(p^{(t)}(\sigma) - \mathbf{e}_{y^{(t)}(\sigma)}\right)
\left(\mathbf{h}^{(t)}_\sigma\right)^\top,
\end{equation}
where $p^{(t)}(\sigma)=\mathrm{softmax}(\mathbf{z}^{(t)}(\sigma))$ and $\mathbf{e}_y$ is one-hot vector of $y$.

To isolate the instability caused by the action-space semantics (which determines where supervision is applied in $\mathbf{W}$),
we analyze the \emph{label-dependent component} of the output-layer gradient:
\begin{equation}
\label{eq:label_component}
\mathbf{G}^{(t)}_{\mathrm{label}}(\sigma)
\triangleq
-\mathbf{e}_{y^{(t)}(\sigma)}\left(\mathbf{h}^{(t)}_\sigma\right)^\top.
\end{equation}
We measure the variance of a random matrix $\mathbf{X}$ by the squared Frobenius deviation:
\begin{equation}
\label{eq:matrix_var_def}
\mathrm{Var}_\sigma(\mathbf{X})
\triangleq
\mathbb{E}_\sigma\!\left[\left\|\mathbf{X}-\mathbb{E}_\sigma[\mathbf{X}]\right\|_F^2\right].
\end{equation}

\subsection{Semantic Inconsistency and Output-Layer Variance}

In many existing list-wise reranking architectures~\cite{pei2019prm,zhang2025goalrank}, the action space is defined over local candidate indices, with $A=N$.
Under this formulation, the semantic meaning of each output dimension depends on the input ordering: the same output row of $\mathbf{W}$ corresponds to different items across different permutations.

Consider a fixed generation step $t$ and a fixed semantic target item $r_t^*\in\mathcal{C}$.
Under a permutation $\sigma$, the supervised label is the position of this item in the permuted list:
\begin{equation}
\label{eq:local_label}
y^{(t)}_\mathrm{loc}(\sigma)
=
\operatorname{pos}_{\sigma}(r_t^*)
\in \{1,\dots,N\}.
\end{equation}
Equivalently, define the indicator variable for whether the target item appears at position $j$:
\begin{equation}
\label{eq:Ij_def}
I^{(t)}_j(\sigma)\triangleq \mathbb{I}\!\left(y^{(t)}_\mathrm{loc}(\sigma)=j\right).
\end{equation}
When $\sigma$ is sampled uniformly over permutations, $I^{(t)}_j$ is Bernoulli with
\begin{equation}
\label{eq:Ij_prob}
\mathbb{P}\!\left(I^{(t)}_j=1\right)=\frac{1}{N}.
\end{equation}

We now focus on the label-dependent update received by a \emph{fixed output parameter row}.
Let $\mathbf{w}_j^\top$ be the $j$-th row of $\mathbf{W}$.
From Eq.~\eqref{eq:label_component}, the label-dependent gradient contribution to $\mathbf{w}_j$ is
\begin{equation}
\label{eq:gj_label_local}
\mathbf{g}^{(t)}_{j,\mathrm{loc}}(\sigma)
\triangleq
-\mathbb{I}\!\left(y^{(t)}_\mathrm{loc}(\sigma)=j\right)\mathbf{h}^{(t)}_\sigma
=
- I^{(t)}_j(\sigma)\,\mathbf{h}^{(t)}_\sigma.
\end{equation}
This term has an ``on--off'' form: the row $\mathbf{w}_j$ receives supervision for the fixed target item $r_t^*$ only when the target happens to be at position $j$. Here, $\mathbf{g}_j(\sigma)$ denotes the gradient vector associated with the
$j$-th row of $\mathbf{W}$ (up to transpose).
For notational simplicity, we treat it as a vector in $\mathbb{R}^d$,
since all subsequent analysis depends only on its Euclidean norm.

\begin{proposition}[Irreducible mapping-induced variance under local indices]
\label{prop:mapping_variance_local}
Fix a step $t$ and a target item $r_t^*$.
Assume $\sigma$ is uniform over permutations so that $\mathbb{P}(I^{(t)}_j=1)=1/N$.
Let
$
\boldsymbol{\mu}^{(t)}_j \triangleq \mathbb{E}_\sigma\!\left[\mathbf{h}^{(t)}_\sigma \mid I^{(t)}_j=1\right].
$
If $\boldsymbol{\mu}^{(t)}_j\neq \mathbf{0}$, then the label-dependent gradient to row $\mathbf{w}_j$ has strictly positive variance:
\begin{equation}
\label{eq:local_lower_bound}
\mathrm{Var}_\sigma\!\left(\mathbf{g}^{(t)}_{j,\mathrm{loc}}\right)
\;\ge\;
\frac{1}{N}\!\left(1-\frac{1}{N}\right)\!\left\|\boldsymbol{\mu}^{(t)}_j\right\|_2^2
\;>\;0.
\end{equation}
\end{proposition}

\begin{proof}
We apply the law of total variance conditioning on $I^{(t)}_j$:
\begin{equation}
\begin{aligned}
\label{eq:total_var_local}
\mathrm{Var}_\sigma\!\left(\mathbf{g}^{(t)}_{j,\mathrm{loc}}\right)
& =
\mathbb{E}_{I_j^{(t)}}\!\left[\mathrm{Var}\!\left(\mathbf{g}^{(t)}_{j,\mathrm{loc}} \mid I^{(t)}_j\right)\right]
+
\mathrm{Var}_{I_j^{(t)}}\!\left(\mathbb{E}\!\left[\mathbf{g}^{(t)}_{j,\mathrm{loc}} \mid I^{(t)}_j\right]\right)\\
& =
\frac{1}{N}\!\left[\mathrm{Var}\!\left(\mathbf{h}^{(t)}_\sigma \mid I^{(t)}_j=1\right)\right]
+
\mathrm{Var}_{I_j^{(t)}}\!\left(\mathbb{E}\!\left[\mathbf{g}^{(t)}_{j,\mathrm{loc}} \mid I^{(t)}_j\right]\right)\\
&\geq \mathrm{Var}_{I_j^{(t)}}\!\left(\mathbb{E}\!\left[\mathbf{g}^{(t)}_{j,\mathrm{loc}} \mid I^{(t)}_j\right]\right)
\end{aligned}
\end{equation}
From Eq.~\eqref{eq:gj_label_local}, we have
\[
\mathbb{E}\!\left[\mathbf{g}^{(t)}_{j,\mathrm{loc}} \mid I^{(t)}_j=0\right]=\mathbf{0},
\qquad
\mathbb{E}\!\left[\mathbf{g}^{(t)}_{j,\mathrm{loc}} \mid I^{(t)}_j=1\right]=-\boldsymbol{\mu}^{(t)}_j.
\]
Therefore, $\mathbb{E}\!\left[\mathbf{g}^{(t)}_{j,\mathrm{loc}} \mid I^{(t)}_j\right]$ is a Bernoulli-driven random vector that equals $-\boldsymbol{\mu}^{(t)}_j$ with probability $p=1/N$ and equals $\mathbf{0}$ with probability $1-p$.
This mapping-induced variance equals
\begin{equation}
\label{eq:mapping_induced_variance}
\mathrm{Var}_{I_j^{(t)}}\!\left(\mathbb{E}\!\left[\mathbf{g}^{(t)}_{j,\mathrm{loc}} \mid I^{(t)}_j\right]\right)
=
\frac{1}{N}\left(1-\frac{1}{N}\right)\left\|\boldsymbol{\mu}^{(t)}_j\right\|_2^2>0.
\end{equation}
Combining with Eq.~\eqref{eq:total_var_local} yields Eq.~\eqref{eq:local_lower_bound}.
\end{proof}

\paragraph{Interpretation.}
Proposition~\ref{prop:mapping_variance_local} states that, regardless of the target item index being fixed or not, a fixed output row $\mathbf{w}_j$ receives an unstable supervision signal because the mapping from items to output rows depends on $\sigma$.
Additionally, this mapping-induced variance does not disappear even when $\mathbf{h}^{(t)}_\sigma$ is invariant to $\sigma$, and it remains strictly positive as long as the target item can appear at multiple positions under the permutation distribution.

\subsection{Stability Under Global Item Spaces}

We now consider an alternative formulation in which the action space is defined over a fixed global item space.
Each item $v\in\mathcal{V}$ is assigned a unique identifier $\mathrm{ID}(v)$, and the supervised label becomes invariant to permutations of $\mathcal{C}$:
\begin{equation}
\label{eq:global_label}
y^{(t)}_{\mathrm{glo}}=\mathrm{ID}(r_t^*).
\end{equation}
Since each row of $\mathbf{W}$ now corresponds to a global label, the label-dependent gradient contribution to the target row is stable with
\begin{equation}
\label{eq:gj_label_global}
\mathbf{g}^{(t)}_{\mathrm{glo}}(\sigma)
\triangleq
-\mathbf{h}^{(t)}_\sigma,
\end{equation}
which is independent of how the candidate list is permuted.

\begin{corollary}[Zero mapping-induced variance under global labels]
\label{cor:global_zero_mapping}
Fix a step $t$ and a target item $r_t^*$.
Under global labels in Eq.~\eqref{eq:global_label}, the mapping-induced variance is zero, and only the encoding variance remains for the corresponding $\sigma$. In particular,
\begin{equation}
\label{eq:global_var}
\mathrm{Var}_\sigma\!\left(\mathbf{g}^{(t)}_{\mathrm{glo}}\right)
=
\mathrm{Var}_\sigma\!\left(\mathbf{h}^{(t)}_\sigma\right).
\end{equation}
\end{corollary}
Ideally, the encoder considers the candidates as an unordered set and $\mathbf{h}^{(t)}$ becomes agnostic to the permutation $\sigma$, then Eq.\eqref{eq:global_var} reduces to zero.
\paragraph{Discussion.}
Comparing Eq.~\eqref{eq:local_lower_bound} and Eq.~\eqref{eq:global_var}, the local position-dependent formulation incurs an additional variance term that comes from the stochastic assignment between items and output rows.
Anchoring supervision in a globally consistent item space removes this mapping-induced variance and leaves only the variance that is intrinsic to the encoder representation $\mathbf{h}^{(t)}_\sigma$.
In the next section, we describe how \name~realizes globally anchored supervision efficiently at scale while still selecting from the local candidate set.

\section{Methodology}

\begin{figure*}[ht]
    \centering
    \includegraphics[width=\linewidth]{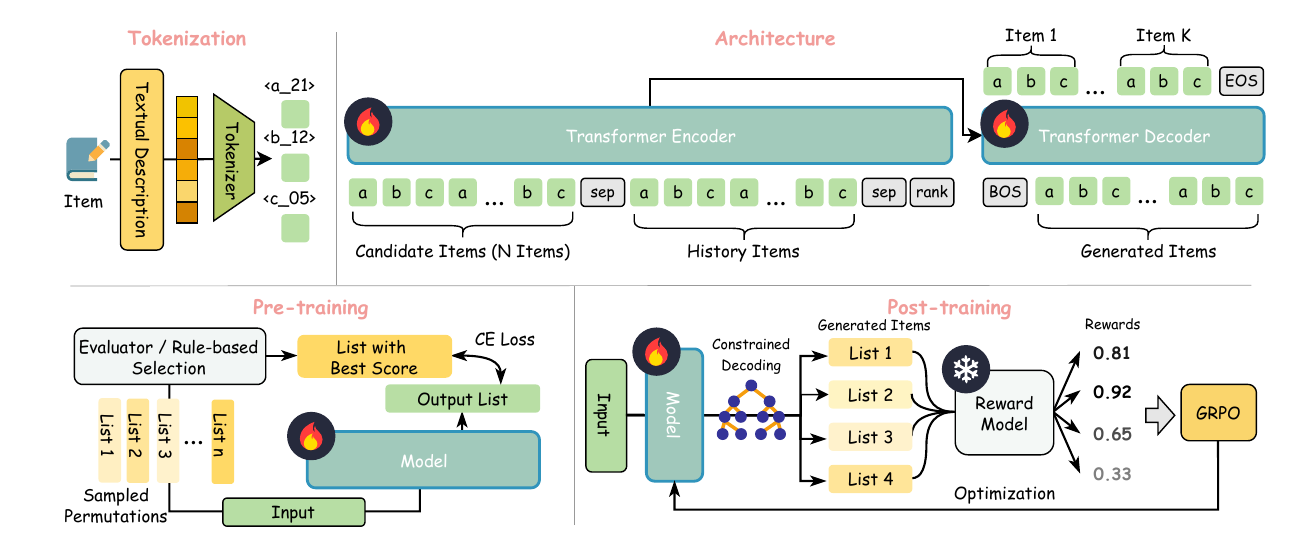}
    \caption{Overview of \name.}
    \label{fig:overview}
\end{figure*}

In this section, we present \name, a novel generative framework designed to resolve the semantic inconsistency of position-dependent action spaces by shifting the reranking paradigm from local index selection to global identifier generation.
Figure~\ref{fig:overview} illustrates the overall framework.
Specifically, we first introduce the tokenization mechanism in Section~\ref{sec:tokenization}, which efficiently constructs the stable global action space by transforming items into Semantic IDs (SIDs). 
Next, in Section~\ref{sec:arc}, we elaborate on the backbone architecture, which reformulates list-wise reranking as a sequence generation task to capture complex inter-item dependencies. 
We then present our decoupled optimization pipeline in Section~\ref{sec:opt}, comprising a supervised pre-training phase (Section~\ref{sec:pre-train}) that initializes the policy with high-quality demonstrations, and a reward-driven post-training phase (Section~\ref{sec:post-train}) that directly maximizes list-wise utility. 
Finally, in Section~\ref{sec:infer}, we detail the prefix-tree constrained decoding algorithm, which effectively bridges the gap between the global action space and local candidate sets, ensuring that generated outputs strictly adhere to candidate constraints.

\subsection{Tokenization} \label{sec:tokenization}

To materialize the concept of a global item space within a reranking framework, we require an item representation scheme that is both globally consistent and computationally efficient. 
Traditional atomic IDs, which assign a unique ID to every item, suffer from high dimensionality, lacking of semantic continuity, and insufficient representation learning for cold-start IDs~\cite{rajput2023recommender}.
These limitations make atomic item IDs unsuitable for modeling the large-scale global vocabulary required by our approach.
To overcome this, we adopt RQ-based Semantic IDs~\cite{zhou2025onerecv2,cuturi2013sinkhorn} to map items into a compact, hierarchical discrete space. 
As illustrated in the tokenization module of Figure~\ref{fig:overview}, this process consists of two key steps: semantic serialization and residual quantization.

\subsubsection{\textbf{Semantic Serialization.}}
To capture the semantics of an item $v \in \mathcal{V}$, we aggregate its textual attributes---such as title and description---into a unified text sequence $T_v$ using a predefined prompt template. This sequence is then fed into a pre-trained text encoder (e.g., Sentence-T5~\cite{ni2022sentence}, Qwen3-Embedding~\cite{qwen3embedding}) to obtain a dense semantic vector $\mathbf{h}_v = \text{Encoder}(T_v)$,
where $\mathbf{h}_v \in \mathbb{R}^D$ encapsulates the global semantic information of the item and $D$ represents the dimension of embeddings.

\subsubsection{\textbf{Residual Quantization.}} 
To transform this continuous vector into discrete tokens suitable for generative modeling, we employ the Residual Quantization (RQ) technique (specifically RQ-VAE~\cite{rajput2023recommender} or RQ-Kmeans~\cite{deng2025onerec,zhou2025onerec}). 
This method approximates the dense vector $\mathbf{h}_v$ through a multi-stage quantization process. 
At each stage $m$ (where $1 \le m \le M$, $M$ is the number of codebooks), the algorithm selects a code from a learned codebook $\mathcal{C}_m$ that minimizes the residual error from the previous stage. 
Formally, let $\mathcal{Q}(\cdot)$ denote the quantization function:
\begin{equation}
    s_v = [c_1, c_2, \dots, c_M] = \mathcal{Q}(\mathbf{h}_v; \Theta_{\mathcal{Q}}),
\end{equation}
where $s_v$ is the corresponding SID sequence that consists of $M$ discrete tokens, and $\Theta_{\mathcal{Q}}$ represents the codebook parameters. 
Each token $c_m$ corresponds to a semantic cluster at the $m$-th level of granularity. 
The resulting tokenizer ensures that the item $v$ is uniquely identified by the sequence $s_v$, allowing us to reformulate item selection as a sequence generation task over a fixed, compact global vocabulary $\mathcal{V}_\text{token} = \bigcup_{m=1}^M \mathcal{C}_m$.

\subsection{Architecture} \label{sec:arc}
To effectively capture the complex dependencies between items as well as patterns in the generated lists, \name~employs a transformer-based encoder-decoder architecture as its backbone to encode the entire input as a sequence and capture inter-item dependencies during generation, following previous approaches~\cite{rajput2023recommender,ju2025generative}.

\subsubsection{\textbf{Input Construction.}} 
The input to the model is a serialized sequence of SIDs constructed using a specific prompt template. Let $S(\cdot)$ denote the serialization function that maps a set of items to a concatenated SID sequence. Given the candidate set $\mathcal{C}$ and the user's interaction history $\mathcal{H}$, the input sequence $X_{\text{input}}$ is formulated as ``$S(\mathcal{C}) \texttt{<sep>} S(\mathcal{H}) \texttt{<sep>} \texttt{<rank>}$''
\noindent where \texttt{<sep>} is a separator token distinguishing different context segments, and \texttt{<rank>} is a prompt token that signals the model to generate reranked lists. 

\subsubsection{\textbf{Autoregressive Generation in Global Item Space.}}
The encoder first processes $X_{\text{input}}$ into a sequence of continuous hidden states. Subsequently, the decoder generates the output sequence autoregressively. A critical distinction of \name~is that the generation occurs over the fixed global vocabulary $\mathcal{V}_\text{token}$, rather than the selection of input indices in $\mathcal{C}$.

At each time step $t$, the decoder predicts the probability distribution of the next token $y_t$ conditioned on the input $X_{\text{input}}$ and previously generated tokens $y_{<t}$:
\begin{equation}
    P(y_t | y_{<t}, X_{\text{input}}) = \text{Softmax}(W_o \cdot h_t),
\end{equation}
where $h_t$ is the decoder's hidden state at step $t$, and $W_o$ projects the state to the vocabulary size. The model continues to generate tokens until the maximum list length is reached. 
As we will describe in Section~\ref{sec:infer}, the resulting sequence of tokens is then reshaped into item-level SIDs to form the final ordered list $R$ through a constrained decoding method during inference.

\subsection{\textbf{Optimization}} \label{sec:opt}

As illustrated in Figure~\ref{fig:overview}, our optimization pipeline is decoupled into two phases: (1) a supervised pre-training phase, which teaches the model to generate high-quality demonstrations; and (2) a reinforcement-based post-training phase, which further maximizes the list-wise utility through direct alignment with reward signals.

\subsubsection{\textbf{Pre-training}} \label{sec:pre-train}

The primary objective of the pre-training phase is to help the model understand item semantics and basic reranking principles. 
A straightforward approach would be to train the model to reconstruct the historical lists exposed to users in the logs~\cite{liu2022rerankingsurvey}. 
However, this approach heavily relies on the quality of the exposed list in the logged data, which may not be the optimal permutation.
This means that reranking models, different from conventional retrieval models, require counterfactual knowledge about unseen lists and learn to distinguish between them.

To mitigate this issue and provide high-quality supervision, we employ a strategy based on offline evaluation. 
Specifically, for each training instance consisting of user history $\mathcal{H}$ and a candidate set $\mathcal{C}$, 
we first generate a set of candidate lists $\mathcal{P} = \{R_1, R_2, \dots, R_L\}$ by sampling from the permutation space of $\mathcal{C}$. 
We then employ a proxy scoring function $E$—which can be a list-wise evaluator or heuristic rules—to assign a score to each permutation. The permutation with the highest score is selected as the target (i.e., $Y^* = \operatorname*{arg\,max}_{R \in \mathcal{P}} E(R, \mathcal{H})$).
With the constructed target sequence $Y^*=[y^*_1, y^*_2, \dots, y^*_{|Y^*|}]$ as a sequence of discrete tokens, 
we optimize \name~using the negative log-likelihood minimization objective in standard Next Token Prediction (NTP) tasks:
\begin{equation}
    \mathcal{L}_{\text{pre}} = - \sum_{t=1}^{|Y^*|} \log P(y^*_t | y^*_{<t}, X_{\text{input}}; \theta).
\end{equation}

\subsubsection{\textbf{Post-training}} \label{sec:post-train}

While the pre-training phase endows the model with basic semantic understanding and reranking capabilities, it remains bounded by the quality of target labels in the offline data. 
To transcend these limitations and directly maximize the collective utility of the generated lists, we introduce a reinforcement learning (RL) based post-training phase. Specifically, we leverage Group Relative Policy Optimization (GRPO)~\cite{shao2024deepseekmath} in our framework, which serves to align the model's global generation capability with the specific objective of maximizing list-wise utility.

\subsection{\textbf{Inference with Constrained Decoding.}} \label{sec:infer}

During post-training and inference, \name~generates the reranked list autoregressively. To ensure that the generated items strictly belong to the current candidate set $\mathcal{C}$ without duplication, we employ a Trie-based constrained decoding strategy~\cite{DBLP:journals/corr/abs-2510-24431}. 
Specifically, for each request, we construct a prefix tree (Trie) from the SIDs of the candidate items. At each decoding step, we constrain the output probability distribution by masking out invalid tokens that do not appear in the Trie's valid branches. Once an item is fully generated, we dynamically prune its corresponding path from the Trie to prevent repetition.
Since the candidate set is typically much smaller than the entire item set (i.e., $N \ll |\mathcal{V}|$) in reranking scenarios, the overhead of Trie construction and masking is negligible.

\section{Experiments}

Through both offline benchmarks and online industrial deployments, we aim to answer the following research questions:
\begin{itemize}[leftmargin=*]
    \item \textbf{RQ1:} How does \name~compare against state-of-the-art reranking baselines in standard offline evaluation settings?
    \item \textbf{RQ2:} Can \name~maintain its effectiveness when generating unseen lists under counterfactual evaluation?
    \item \textbf{RQ3:} Is the proposed decoupled two-stage optimization pipeline essential for achieving optimal performance?
    \item \textbf{RQ4:} How do variations in key hyperparameters influence the final performance of the model?
    \item \textbf{RQ5:} Is the global action space necessary in modern generative recommendation frameworks?
    \item \textbf{RQ6:} Does the model confer specific advantages in handling cold-start scenarios compared to traditional approaches?
    \item \textbf{RQ7:} Is \name~effective when deployed in a real-world, large-scale industrial production system?
\end{itemize}

\begin{table}
\centering
\caption{Dataset statistics.}
\label{tab:dataset_statistics}
\resizebox{0.9\linewidth}{!}{
\begin{tabular}{lcccc} 
\toprule
Dataset            & Users   & Items  & Interactions & Lists    \\ 
\midrule
MovieLens-1M       & 6,020   & 3,043  & 995,154      & 161,646  \\
Amazon Books       & 35,732  & 38,121 & 1,960,674    & 311,386  \\
Industrial Dataset & 200,775 & 17,014 & 5,173,698    & 589,751  \\
\bottomrule
\end{tabular}}
\end{table}

\begin{table*}[ht]
\centering
\caption{Overall experiments. The best results are highlighted in bold, and the second-best results are underlined. Rel. Improvement denotes the percentage improvement of the proposed method relative to the best-performing baseline.}
\label{tab:overall}
\resizebox{\linewidth}{!}{
\begin{tabular}{clcccccccccccc} 
\toprule
\multicolumn{2}{c}{\multirow{2}{*}{Methods}}                                                      & \multicolumn{4}{c}{Amazon Books}                                  & \multicolumn{4}{c}{MovieLens-1M}                                  & \multicolumn{4}{c}{Industrial Dataset}                             \\ 
\cmidrule{3-14}
\multicolumn{2}{c}{}                                                                              & Precision      & NDCG           & MAP            & F1             & Precision      & NDCG           & MAP            & F1             & Precision      & NDCG           & MAP            & F1              \\ 
\midrule
\multirow{4}{*}{\begin{tabular}[c]{@{}c@{}}Generator-\\Only\end{tabular}}      & DNN~\cite{covington2016youtube}              & 60.28          & 69.61          & 58.58          & 62.45          & 56.86          & 70.30          & 59.28          & 62.16          & 32.80          & 51.26          & 41.02          & 38.97           \\
& DLCM~\cite{ai2018learning}             & 66.80          & 75.88          & 65.39          & 69.28          & 62.31          & 73.87          & 63.82          & 67.96          & 47.86          & 73.03          & 63.16          & 56.89           \\
& PRM~\cite{pei2019prm}              & 67.86          & 76.88          & 66.44          & 70.42          & 60.09          & 72.85          & 62.21          & 65.51          & 47.89          & 67.09          & 54.59          & 56.88           \\
& GoalRank~\cite{zhang2025goalrank}         & \uline{80.35}  & \uline{84.88}  & \uline{77.91}  & \uline{83.44}  & \uline{73.56}  & \uline{83.43}  & \uline{76.16}  & \uline{80.15}  & \uline{60.16}  & \uline{88.08}  & \uline{82.06}  & \uline{71.48}   \\ 
\midrule
\multirow{5}{*}{\begin{tabular}[c]{@{}c@{}}Generator-\\Evaluator\end{tabular}} & PIER~\cite{shi2023pier}             & 71.14          & 80.22          & 71.62          & 73.74          & 62.74          & 75.99          & 65.98          & 68.74          & 56.69          & 80.15          & 70.52          & 67.52           \\
   & NAR4Rec~\cite{ren2024nar4rec}              & 70.08          & 79.46           & 70.69          & 72.66          & 62.81           & 75.01         & 65.42          & 68.31            & 51.93          & 65.16          & 53.13          & 61.56           \\
   & MG-E-3~\cite{yang2025mgeclig}              & 68.76          & 76.36          & 65.82          & 71.33          & 55.51          & 67.39          & 55.52          & 55.51          & 48.53          & 67.38          & 54.89          & 57.67           \\
   & MG-E-20~\cite{yang2025mgeclig}             & 72.99          & 78.68          & 68.66          & 75.72          & 58.66          & 69.86          & 58.60          & 64.18          & 51.58          & 69.29          & 57.26          & 61.08           \\
   & MG-E-100~\cite{yang2025mgeclig}            & 77.21          & 82.15          & 73.78          & 80.09          & 60.64          & 70.97          & 59.93          & 66.29          & 53.21          & 70.67          & 58.80          & 63.06           \\ 
\midrule \rowcolor{results}
   & \name          & \textbf{83.75} & \textbf{90.08} & \textbf{85.10} & \textbf{86.97} & \textbf{75.79} & \textbf{87.56} & \textbf{81.19} & \textbf{82.57} & \textbf{62.56} & \textbf{90.15} & \textbf{84.84} & \textbf{74.32}  \\
   \rowcolor{results}
   & Rel. Improvement & $\uparrow4.23\%$         & $\uparrow6.13\%$         & $\uparrow9.23\%$         & $\uparrow4.23\%$         & $\uparrow3.03\%$         & $\uparrow4.95\%$         & $\uparrow6.60\%$         & $\uparrow3.02\%$         & $\uparrow3.99\%$         & $\uparrow2.35\%$        & $\uparrow3.39\%$        & $\uparrow3.97\%$          \\
\bottomrule
\end{tabular}}
\end{table*}

\subsection{Dataset and Evaluation Metrics}

\textbf{Datasets.}
To comprehensively evaluate the performance of \name, we conduct experiments on three datasets, including two widely used public benchmarks (following previous work~\cite{zhang2025goalrank,mao2025denoising}), Amazon Books~\cite{mcauley2015image} and MovieLens-1M~\cite{harper2015movielens}, and a large-scale Industrial dataset collected from a real-world e-commerce platform. Amazon Books and MovieLens-1M serve as standard testbeds to ensure reproducibility, while the Industrial dataset is employed to verify the effectiveness and robustness of our method in a practical production environment. The detailed statistics of these datasets are summarized in Table~\ref{tab:dataset_statistics}.

\noindent\textbf{Evaluation Metrics.}
We adopt standard list-wise metrics to assess the quality of the reranked lists, specifically \texttt{Precision@K}, \texttt{NDCG@K}, \texttt{MAP@K}, and \texttt{F1@K}. These metrics evaluate the recommendation performance from various perspectives, including accuracy and ranking position. 

\subsection{Baselines}

To evaluate the effectiveness of \name, we compare it against a diverse set of representative baselines. We categorize these methods into Generator-only approaches and Generator-Evaluator frameworks.
\textbf{(1) Generator-only Baselines:} these methods employ a single model to directly generate the final ranked list, either through point-wise scoring or sequential generation, including \textbf{DNN}~\cite{covington2016youtube}, \textbf{DLCM}~\cite{ai2018learning}, \textbf{PRM}~\cite{pei2019prm}, \textbf{GoalRank}~\cite{zhang2025goalrank}. \textbf{(2) Generator-Evaluator Baselines:} these frameworks decouple the reranking process into two phases: generators first produce a set of candidate permutations, and a separate evaluator subsequently selects the optimal sequence, including \textbf{PIER}~\cite{shi2023pier}, \textbf{NAR4Rec}~\cite{ren2024nar4rec}, \textbf{MG-E}~\cite{yang2025mgeclig}.

\subsection{Implementation Details}

To construct the SIDs, we first generate dense semantic vectors by serializing item metadata into a unified text sequence. For the Amazon Books dataset, we concatenate the \texttt{title}, \texttt{categories}, \texttt{description}, \texttt{price}, and \texttt{brand} fields. For MovieLens-1M, we utilize the \texttt{title} and \texttt{genres}. These sequences are encoded using the pre-trained Qwen3-embedding-4B~\cite{qwen3embedding} model to obtain high-quality item representations. For the quantization process, we employ the RQ-Kmeans~\cite{zhou2025onerec} algorithm with a hierarchical depth of $M=4$ and the codebook size $256$.
Our model backbone is built upon the Transformer-based Encoder-Decoder architecture, following the configuration of standard T5~\cite{2020t5} models. We use the AdamW~\cite{loshchilov2017decoupled} optimizer for both training stages. 
In the supervised pre-training phase, we set the learning rate to $5\text{e}^{-4}$ with a weight decay of $0.01$. 
In the post-training phase, we reduce the learning rate to $5\text{e}^{-6}$. For the GRPO process, we set the sampling temperature $\tau=1.0$ and the group size $20$ to balance exploration and gradient estimation variance. 
During the inference phase, we employ beam search decoding with a beam size of $20$. Following previous work~\cite{zhang2025goalrank,mao2025denoising}, we configure the task to generate a list of length 6 from a candidate set of 50 items. Accordingly, the cutoff $K$ for all evaluation metrics is set to 6.

\subsection{Overall Experiments (RQ1)}

To demonstrate the superior performance of \name, we compare it against state-of-the-art baselines across three datasets under the standard offline evaluation setting. 
We utilize the ground-truth evaluation metrics (specifically, the arithmetic mean of the target metrics) directly as the reward signal during the post-training phase. The comparative results are summarized in Table~\ref{tab:overall}, from which we draw the following key observations: (1) \name~consistently outperforms all baselines across all datasets and metrics. This validates that anchoring the reranking process in a stable Global Item Space allows the model to learn a more robust and transferable valuation standard. (2) Generally, Generator-Evaluator methods outperform traditional Generator-only approaches (e.g., DLCM, PRM). This trend supports the prevailing view that decoupling generation and evaluation enables a more exhaustive exploration of the permutation space, as the evaluator can explicitly model intra-list correlations that generators may miss.
(3) GoalRank stands out as a formidable exception among Generator-only methods. By employing a group-relative optimization strategy similar to the one used in our post-training phase, GoalRank surpasses standard Generator-Evaluator methods. This indicates that advanced policy optimization can effectively bridge the gap between Generator-only and Generator-Evaluator paradigms. Nonetheless, \name~still outperforms GoalRank, highlighting that the benefits of our framework stem not just from the optimization method, but from the global space formulation.
(4) Regarding the MG-E variants, we observe that performance steadily improves as the number of generators increases (from MG-E-3 to MG-E-100). This confirms that expanding the search space via ensemble generation leads to better ranking quality. Nevertheless, \name~achieves superior results with a single generative policy, bypassing the high computational overhead of running hundreds of parallel generators.

\subsection{Counterfactual Experiments (RQ2)}

Standard offline evaluation relies on static historical logs, which inherently suffer from selection bias—we can only evaluate lists that were previously exposed by the logging policy. To rigorously assess \name~and other generative solutions in a more realistic environment, we employ a simulation-based evaluation on the MovieLens-1M dataset, following established protocols~\cite{zhao2023kuaisim}, in addition to the standard offline evaluation.
Figure~\ref{fig:sim_box} presents the box plot of the reward scores achieved by different models in the simulator. We observe the following: (1) \name~exhibits a significant performance advantage over all baselines. Its score distribution is notably higher, with both the median and the upper quartile far exceeding those of the strongest baselines. This indicates that \name~is not only better on average but also possesses a higher ceiling for optimizing list utility. (2) Results are consistent with our offline evaluation. GoalRank and MG-E-100 remain the most competitive baselines. Their relatively strong performance confirms that modeling list-wise interactions is crucial. However, they still fall short of \name, highlighting the limitations of position-dependent action spaces in exploring optimal policies.

\subsection{Ablation Study (RQ3)}

\begin{figure}
    \centering
    \includegraphics[width=\linewidth]{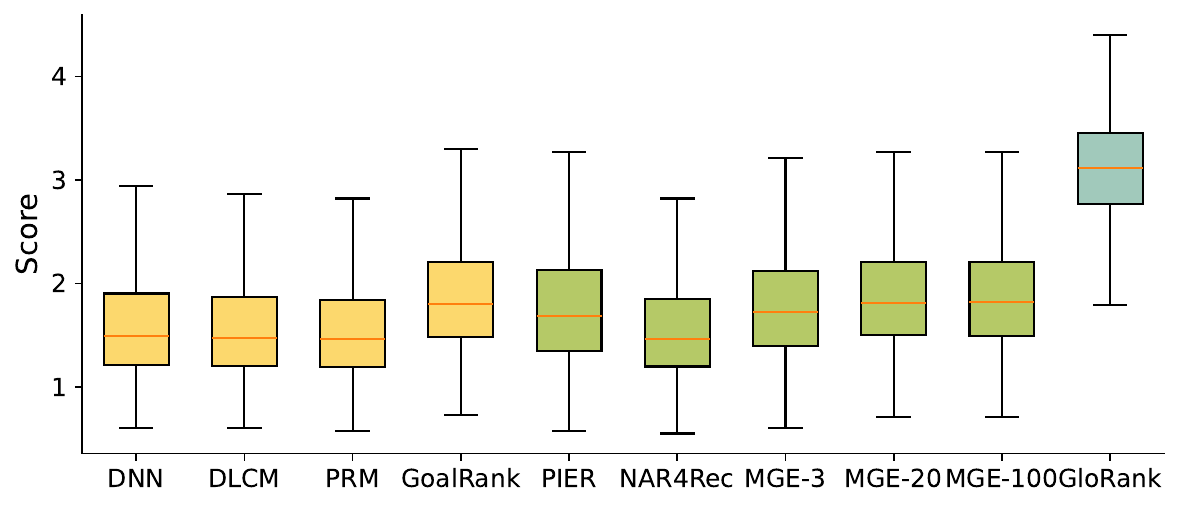}
    \caption{Score distribution comparison.}
    \label{fig:sim_box}
\end{figure}

\begin{figure*}
    \centering
    \includegraphics[width=0.8\linewidth]{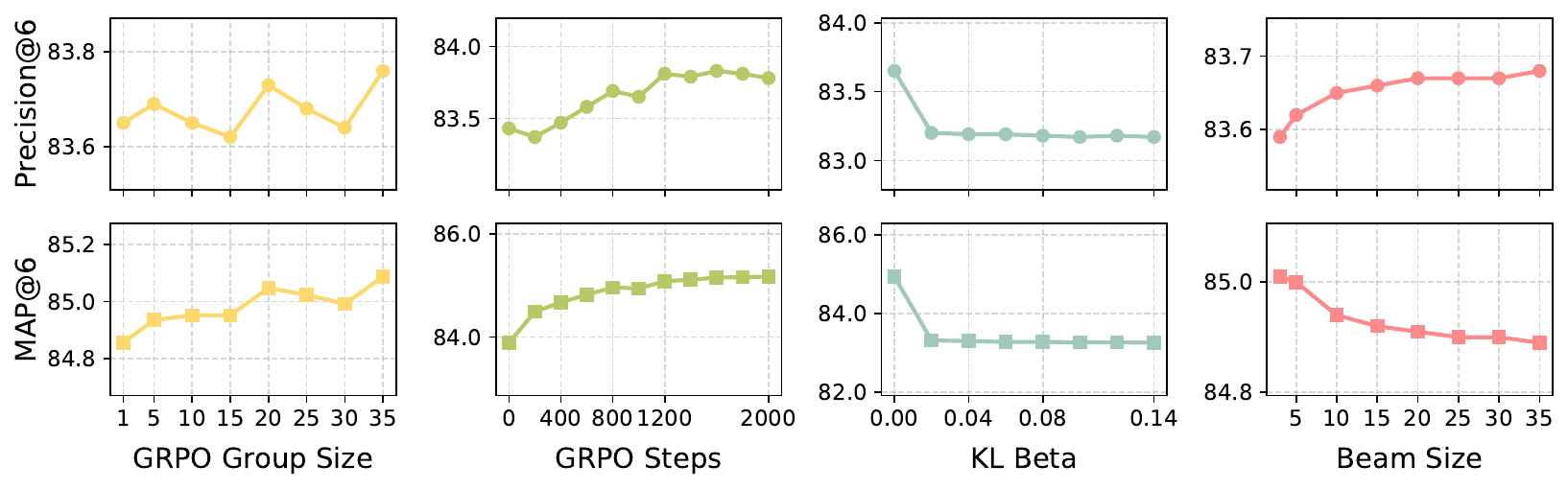}
    \caption{Hyperparameter analysis.}
    \label{fig:hyperparameter}
\end{figure*}

To verify the necessity and effectiveness of the proposed decoupled two-stage optimization pipeline, we conduct an ablation study by comparing \name~with two variants:
\begin{itemize}[leftmargin=*]
    \item \textbf{w/o pre:} This variant skips the supervised pre-training phase and directly initializes the model with random weights for the reward-driven post-training.
    \item \textbf{w/o post:} This variant employs only the pre-training phase.
\end{itemize}
The experimental results on Amazon Books and MovieLens-1M are summarized in Table~\ref{tab:ablation}. From these results, we draw the following key observations:
(1) First, both optimization stages contribute positively to the final performance. Removing either stage leads to a degradation across all evaluation metrics on both datasets, confirming that the combination of supervised learning and reward-based utility maximization is essential for an effective reranking framework. (2) Second, the pre-training phase serves as the fundamental cornerstone of \name. As observed, the w/o pre-training variant suffers a catastrophic performance drop (e.g., NDCG declines from 90.08\% to 24.13\% on Amazon Books). This indicates that without the initialization provided by supervised learning on high-quality demonstrations, the model struggles to comprehend the basic semantics of items or the syntax of valid ranking lists. (3) Third, while the post-training phase yields consistent improvements, its impact is particularly pronounced on ranking-sensitive metrics such as NDCG and MAP. For instance, on the MovieLens-1M dataset, post-training boosts MAP by 2.34\% (from 78.85\% to 81.19\%) and NDCG by 1.62\% (from 85.94\% to 87.56\%), compared to smaller gains in Precision and F1. We attribute this phenomenon to the distinct objectives of the two phases: the pre-training phase primarily incentivizes the model to distinguish positive items from the global pool, potentially underemphasizing the precise inter-item ordering within the list. In contrast, the post-training phase facilitates direct interaction with the reward model, providing explicit feedback on list-wise utility.

\subsection{Hyperparameter Analysis (RQ4)}

\begin{table}[t]
\centering
\caption{Ablation study. The best results are highlighted in bold, and the second-best results are underlined.}
\label{tab:ablation}
\resizebox{\linewidth}{!}{
\begin{tabular}{lcccccccc} 
\toprule
\multirow{2}{*}{Methods} & \multicolumn{4}{c}{Amazon Books}                                  & \multicolumn{4}{c}{MovieLens-1M}                                   \\ 
\cmidrule{2-9}
                         & Precision      & NDCG           & MAP            & F1             & Precision      & NDCG           & MAP            & F1              \\ 
\midrule
w/o pre                  & 20.85          & 24.13          & 14.00          & 21.64          & 18.53          & 24.04          & 14.19          & 20.09           \\
w/o post                 & \uline{83.44}  & \uline{89.15}  & \uline{83.68}  & \uline{86.63}  & \uline{75.51}  & \uline{85.94}  & \uline{78.85}  & \uline{82.28}   \\
Ours                     & \textbf{83.75} & \textbf{90.08} & \textbf{85.10} & \textbf{86.97} & \textbf{75.79} & \textbf{87.56} & \textbf{81.19} & \textbf{82.57}  \\
\bottomrule
\end{tabular}}
\end{table}

To assess the robustness of \name~and identify optimal configurations, we analyze the impact of four critical hyperparameters: the group size $G$, the number of training steps, and the KL regularization coefficient $\beta$ during the post-training phase, as well as the beam size $B$ during inference. The experimental results on the Amazon Books dataset are presented in Figure~\ref{fig:hyperparameter}. Note that we conduct experiments across all four metrics and observe that Precision exhibits trends highly consistent with F1, while MAP aligns closely with NDCG. Therefore, due to space limitations, we only present the results for Precision and MAP in this section.

\paragraph{\textbf{Impact of GRPO Configurations.}}
First, regarding the Group Size ($G$), we observe a consistent performance improvement across all metrics as $G$ increases. This aligns with the theoretical foundation of GRPO: a larger group size provides more samples for estimating the group-relative baseline, thereby reducing the variance of the advantage estimator and stabilizing the policy gradients.
Second, regarding Training Steps, the model's performance improves rapidly in the initial stages and eventually reaches a stable plateau. This indicates that \name~converges efficiently without suffering from severe collapse.
Third, regarding the KL Coefficient ($\beta$) in the GRPO objective~\cite{shao2024deepseekmath}, we find that the best performance is achieved when $\beta=0$ (i.e., removing the KL penalty). This observation is consistent with recent papers~\cite{hu2025open,liu2025understanding,yu2025dapo}. 
We attribute this to the fact that strict adherence to the pre-trained reference policy limits the model's exploration, whereas removing the constraint allows the policy to fully optimize for the ranking utility.

\paragraph{\textbf{Impact of Inference Beam Size.}}
The analysis of Beam Size reveals an interesting trade-off between set-based and rank-based metrics. As the beam size increases, \texttt{Precision@6} generally improves, suggesting that a broader search space helps the model identify a more accurate set of relevant items. However, \texttt{MAP@6} exhibits an inverse trend, degrading slightly as the beam size grows. This phenomenon suggests that while larger beams find high-probability sequences that contain more relevant items, the autoregressive probability maximization does not strictly correlate with the optimal ordering required for MAP. 

\subsection{Impact of Global Action Space (RQ5)}

\begin{figure}
    \centering
    \includegraphics[width=1.0\linewidth]{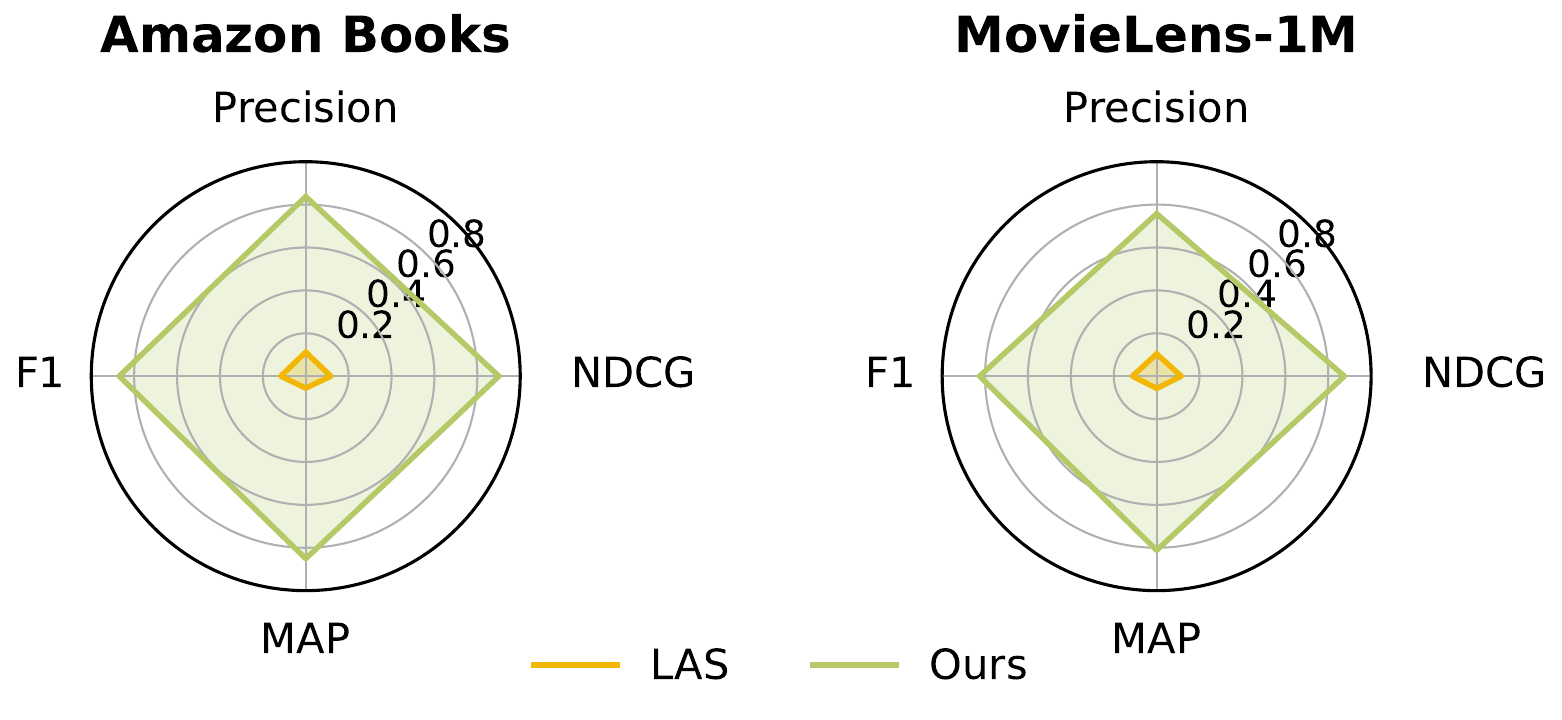}
    \caption{Performance comparison between \name~(Global Action Space) and the local action space variant (LAS).}
    \label{fig:global_vs_local}
\end{figure}

To answer RQ5 and verify the fundamental necessity of the global action space in modern generative frameworks, we compare \name~with a local action space (LAS) variant. 
Instead of generating global SIDs, this variant is restricted to a position-dependent action space, forcing the model to generate local indices (e.g., tokens representing $1$ to $N$) to select items based on their relative positions in the input candidate list. Figure \ref{fig:global_vs_local} visualizes the performance comparison across four metrics on two public datasets. 
The results reveal a substantial performance degradation for the local action space variant.
This is because the local action space introduces noise during training, which hinders optimization and prevents stable convergence.
Furthermore, empirical analysis indicates that the LAS variant tends to memorize high-frequency index combinations rather than selecting appropriate item sequences based on the context, leading to poor performance. This observation provides support for our theoretical analysis on semantic inconsistency.

\subsection{Generalization to Unseen Items (RQ6)}

\begin{figure}
    \centering
    \includegraphics[width=\linewidth]{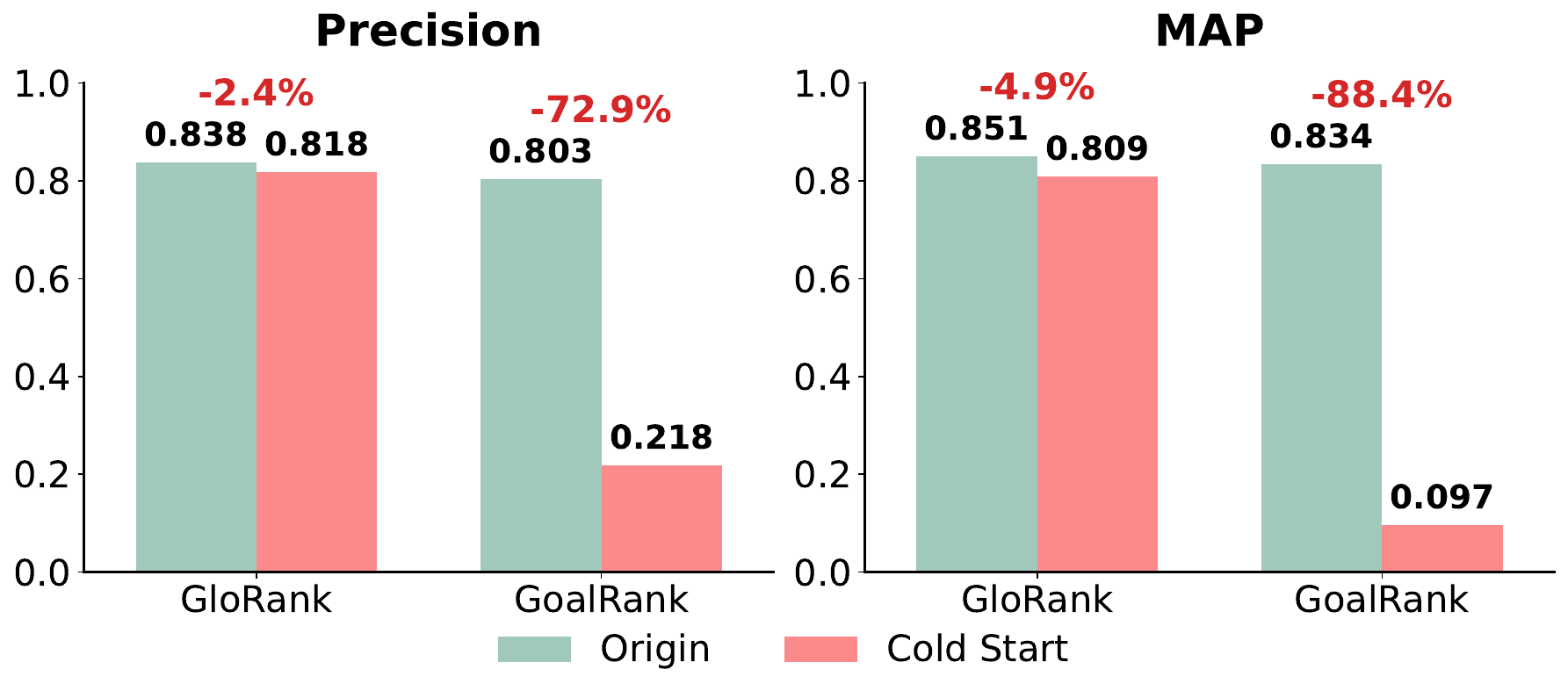}
    \caption{Robustness comparison between \name~and GoalRank under normal and cold-start scenarios.
    }
    \label{fig:cold-start}
\end{figure}

We conduct a specialized experiment to investigate the model's robustness in handling cold-start scenarios—a perennial challenge in recommender systems where new items lack historical interaction data. We simulate a cold-start setting by randomly masking 5\% of the items in the test set from the Amazon Books training set. These items are removed entirely during the training phase but are retained in the test set. Figure~\ref{fig:cold-start} illustrates the performance comparison between \name~and the strongest baseline, GoalRank, under both normal and cold-start settings. From the results, 
we can see that 
\name~demonstrates significantly superior robustness compared to GoalRank. While GoalRank suffers a catastrophic performance drop in the cold-start setting (e.g., MAP decreases by over 88.4\%), \name~maintains a high level of accuracy with only a marginal decline. This validates that the global item space effectively transfers knowledge from seen to unseen items. Even if a specific item ID is novel, its constituent semantic tokens share similarities with known items, allowing the model to infer its utility based on content rather than memorized identity. 
In contrast, 
GoalRank relies on learning distinct embeddings for atomic IDs. When encountering an unseen item, the model must use a randomly initialized embedding, which acts as noise. Crucially, because reranking is a list-wise task that explicitly models inter-item correlations, this noise is not isolated to the unseen item itself. Instead, it propagates through the interactions, corrupting the contextual representation of the entire list and impairing the scoring of even the known items in the same batch.

\subsection{Online Experiments (RQ7)}

\begin{table}
\centering
\caption{Online A/B testing results.All improvements are statistically significant with student t-test $p<0.05$.}
\label{tab:online_res}
\resizebox{\linewidth}{!}{
\begin{tabular}{lccccc} 
\toprule
Metric     & Watch Time & Effective View & Like & Comment & Forward  \\ 
\hline
Imp.  & +0.095\%          & +0.111\%         & +0.286\%   & +0.462\%   & +0.447\%            \\
\bottomrule
\end{tabular}}
\end{table}

To validate the practical effectiveness and robustness of \name~in a real-world production environment, we conduct online A/B testing on a leading content platform. This platform serves over 400 million Daily Active Users (DAU) and manages a corpus of tens of millions of items. 
We allocate  7.8\% of the total live traffic to the experimental group, while a control group of equal size was served by GoalRank, the current state-of-the-art baseline deployed in the production system. The experiment was conducted over a period of 14 days to account for weekly fluctuations and ensure statistical significance. We focus on several core business metrics to evaluate user engagement and platform growth.
As shown in Table~\ref{tab:online_res}, \name~achieved consistent and significant improvements across all monitored metrics compared to the baseline.
Specifically, the increase in Effective Views and Total Time Spent indicates that the reranked lists generated by \name~effectively promote items that trigger both clicks and sustained consumption. Furthermore, we observe a growth in Active Devices by 0.030\%, which suggests that \name~also positively impacts user long-term retention and stickiness.
These results confirm that by modeling the global action space, \name~captures intrinsic user interests more accurately than position-dependent approaches.

\section{Related Work}

\subsection{List-wise Reranking in Recommendations}

Reranking serves as the decisive final phase in the multi-stage recommendation pipeline~\cite{liu2022rerankingsurvey,lin2025gref,meng2025generative,wang2025nlgr,huzhang2021aliexpress,feng2021revisit,xi2022multi}. Formally, the task aims to select and order a target list of size $K$ from a larger candidate set of size $N$ ($K < N$). Unlike the preceding ranking stage that primarily estimates point-wise relevance~\cite{dai2025onepiece,huang2025towards}, reranking explicitly models intra-list correlations, such as diversity and complementarity, to maximize the collective utility of the generated sequence.
Existing neural reranking approaches can be broadly categorized into two paradigms: Generator-only and Generator-Evaluator methods.
Generator-only methods~\cite{ai2018learning,bello2018seq2slate,pei2019prm,liu2023gfn4list,zhang2025goalrank} aim to directly output the final ranked list through a unified model. Representative works like DLCM~\cite{ai2018learning} utilized RNNs to capture sequential dependencies. Subsequently, PRM~\cite{pei2019prm} introduced the transformer architecture for reranking, leveraging self-attention mechanisms to explicitly model the mutual influence among all candidate items. More recently, GoalRank~\cite{zhang2025goalrank} revisited this paradigm from a one-stage perspective, employing group-relative optimization to directly maximize list-wise rewards.
Generator-Evaluator methods~\cite{shi2023pier,ren2024nar4rec,yang2025mgeclig}, in contrast, decouple the task into two distinct phases: a generator first proposes high-quality candidate permutations, and an evaluator selects the best one. PIER~\cite{shi2023pier} employs a SimHash-based generator followed by an omnidirectional context-aware evaluator. NAR4Rec~\cite{ren2024nar4rec} introduces a non-autoregressive generator with contrastive decoding to enhance efficiency. To further expand the search space, MG-E~\cite{yang2025mgeclig} utilizes an ensemble of multiple generators to produce diverse permutations for the evaluator. \name~fundamentally diverges from these approaches by anchoring the reranking process in a Global Item Space. \name~transforms list generation into a token prediction task over a fixed global vocabulary. This paradigm shift allows \name~to apply a consistent, position-invariant estimation to local candidates.

\subsection{Generative Recommendations}

Generative recommendation represents a paradigm shift from the traditional ``retrieve-and-rank'' strategy to a ``next-item generation'' framework. Unlike discriminative models that calculate matching scores for user-item pairs, generative models reformulate recommendation as a sequence generation task, predicting the next target item conditioned on the context history~\cite{li2025survey,chen2024hllm,wang2024learnable,DBLP:conf/icde/ZhengHLCZCW24}. Following recent comprehensive surveys, we discuss the advancements in this field from three decoupled perspectives: tokenization, architecture, and optimization. (1) Tokenization. Tokenization serves as the bridge between discrete items and the generative model. While early sequence models relied on random atomic IDs, modern generative approaches explore more semantic-rich representations. 
P5~\cite{geng2022recommendation} proposes a unified text-to-text framework that represents items using raw text tokens (e.g., titles) via natural language prompts. 
To address the inefficiency of text and the lack of semantics in atomic IDs, TIGER~\cite{rajput2023recommender} introduces Semantic IDs, utilizing RQ-VAE to quantize item embeddings into hierarchical discrete tuples, thereby enabling the model to capture collaborative signals directly within the ID structure.
(2) Architecture. The backbone architecture determines how the model processes item sequences and generates future interactions. The landscape is primarily divided into Encoder-Decoder and Decoder-only structures. 
M6-Rec~\cite{cui2022m6} utilizes a standard Transformer Encoder-Decoder to support multi-modal tasks including retrieval and explanation. 
In contrast, some work favor Decoder-only architectures for their superior scaling properties. For instance, HSTU~\cite{DBLP:conf/icml/ZhaiLLWLCGGGHLS24} designs a high-performance generative architecture specifically optimized for sequential transduction in large-scale recommendation. 
Similarly, LLaRA~\cite{DBLP:conf/sigir/LiaoL0WYW024} adapts the LlaMA architecture to the recommendation domain, bridging the gap between pre-trained language knowledge and sequential user behavior.
(3) Optimization.
Optimization strategies in generative recommendation focus on aligning the model's generation probability with user preferences. Beyond the standard Next Token Prediction (NTP) objective, Instruction Tuning has emerged as a key technique to enhance generalization. 
TallRec~\cite{bao2023tallrec} demonstrates that fine-tuning Large Language Models (LLMs) with recommendation-specific instructions significantly improves performance in few-shot settings. 
Furthermore, researchers are increasingly exploring alignment techniques to ensure the generated recommendations satisfy complex utility constraints beyond simple accuracy. 
In contrast to these works, which primarily focus on next-item generation, \name~targets the reranking setting and reformulates reranking from local index selection to global identifier generation.

\section{Conclusion}

In this paper, we propose \name, a novel generative framework that fundamentally addresses the limitations of existing position-dependent reranking paradigms. By identifying candidate-position coupling as a critical source of semantic misalignment and optimization instability, we advocate for a paradigm shift from local discrimination to global generation. 
Specifically, \name~reformulates list-wise reranking as a sequence generation task over a fixed global vocabulary constrained by the candidates in the decoding tree. This design effectively anchors the scoring mechanism, ensuring consistent item evaluation independent of dynamic input permutations. 
To fully exploit this architecture, we devise a decoupled optimization pipeline that combines supervised pre-training for semantic alignment with GRPO for utility maximization. 
Extensive experiments on two public benchmarks and a large-scale industrial dataset, coupled with online A/B tests, demonstrate that \name~consistently outperforms state-of-the-art baselines.

\clearpage

\balance
\bibliographystyle{ACM-Reference-Format}
\bibliography{reference}


\begin{thebibliography}{50}


\ifx \showCODEN    \undefined \def \showCODEN     #1{\unskip}     \fi
\ifx \showISBNx    \undefined \def \showISBNx     #1{\unskip}     \fi
\ifx \showISBNxiii \undefined \def \showISBNxiii  #1{\unskip}     \fi
\ifx \showISSN     \undefined \def \showISSN      #1{\unskip}     \fi
\ifx \showLCCN     \undefined \def \showLCCN      #1{\unskip}     \fi
\ifx \shownote     \undefined \def \shownote      #1{#1}          \fi
\ifx \showarticletitle \undefined \def \showarticletitle #1{#1}   \fi
\ifx \showURL      \undefined \def \showURL       {\relax}        \fi
\providecommand\bibfield[2]{#2}
\providecommand\bibinfo[2]{#2}
\providecommand\natexlab[1]{#1}
\providecommand\showeprint[2][]{arXiv:#2}

\bibitem[Ai et~al\mbox{.}(2018)]%
        {ai2018learning}
\bibfield{author}{\bibinfo{person}{Qingyao Ai}, \bibinfo{person}{Keping Bi}, \bibinfo{person}{Jiafeng Guo}, {and} \bibinfo{person}{W.~Bruce Croft}.} \bibinfo{year}{2018}\natexlab{}.
\newblock \showarticletitle{Learning a Deep Listwise Context Model for Ranking Refinement}. In \bibinfo{booktitle}{\emph{The 41st International {ACM} {SIGIR} Conference on Research {\&} Development in Information Retrieval, {SIGIR} 2018, Ann Arbor, MI, USA, July 08-12, 2018}}, \bibfield{editor}{\bibinfo{person}{Kevyn Collins{-}Thompson}, \bibinfo{person}{Qiaozhu Mei}, \bibinfo{person}{Brian~D. Davison}, \bibinfo{person}{Yiqun Liu}, {and} \bibinfo{person}{Emine Yilmaz}} (Eds.). \bibinfo{publisher}{{ACM}}, \bibinfo{pages}{135--144}.
\newblock
\href{https://doi.org/10.1145/3209978.3209985}{doi:\nolinkurl{10.1145/3209978.3209985}}


\bibitem[Bao et~al\mbox{.}(2023)]%
        {bao2023tallrec}
\bibfield{author}{\bibinfo{person}{Keqin Bao}, \bibinfo{person}{Jizhi Zhang}, \bibinfo{person}{Yang Zhang}, \bibinfo{person}{Wenjie Wang}, \bibinfo{person}{Fuli Feng}, {and} \bibinfo{person}{Xiangnan He}.} \bibinfo{year}{2023}\natexlab{}.
\newblock \showarticletitle{TALLRec: An Effective and Efficient Tuning Framework to Align Large Language Model with Recommendation}. In \bibinfo{booktitle}{\emph{Proceedings of the 17th {ACM} Conference on Recommender Systems, RecSys 2023, Singapore, Singapore, September 18-22, 2023}}, \bibfield{editor}{\bibinfo{person}{Jie Zhang}, \bibinfo{person}{Li~Chen}, \bibinfo{person}{Shlomo Berkovsky}, \bibinfo{person}{Min Zhang}, \bibinfo{person}{Tommaso~Di Noia}, \bibinfo{person}{Justin Basilico}, \bibinfo{person}{Luiz Pizzato}, {and} \bibinfo{person}{Yang Song}} (Eds.). \bibinfo{publisher}{{ACM}}, \bibinfo{pages}{1007--1014}.
\newblock
\href{https://doi.org/10.1145/3604915.3608857}{doi:\nolinkurl{10.1145/3604915.3608857}}


\bibitem[Bello et~al\mbox{.}(2018)]%
        {bello2018seq2slate}
\bibfield{author}{\bibinfo{person}{Irwan Bello}, \bibinfo{person}{Sayali Kulkarni}, \bibinfo{person}{Sagar Jain}, \bibinfo{person}{Craig Boutilier}, \bibinfo{person}{Ed~Huai{-}hsin Chi}, \bibinfo{person}{Elad Eban}, \bibinfo{person}{Xiyang Luo}, \bibinfo{person}{Alan Mackey}, {and} \bibinfo{person}{Ofer Meshi}.} \bibinfo{year}{2018}\natexlab{}.
\newblock \showarticletitle{Seq2Slate: Re-ranking and Slate Optimization with RNNs}.
\newblock \bibinfo{journal}{\emph{CoRR}}  \bibinfo{volume}{abs/1810.02019} (\bibinfo{year}{2018}).
\newblock
\showeprint[arXiv]{1810.02019}
\urldef\tempurl%
\url{http://arxiv.org/abs/1810.02019}
\showURL{%
\tempurl}


\bibitem[Chen et~al\mbox{.}(2024)]%
        {chen2024hllm}
\bibfield{author}{\bibinfo{person}{Junyi Chen}, \bibinfo{person}{Lu Chi}, \bibinfo{person}{Bingyue Peng}, {and} \bibinfo{person}{Zehuan Yuan}.} \bibinfo{year}{2024}\natexlab{}.
\newblock \showarticletitle{Hllm: Enhancing sequential recommendations via hierarchical large language models for item and user modeling}.
\newblock \bibinfo{journal}{\emph{arXiv preprint arXiv:2409.12740}} (\bibinfo{year}{2024}).
\newblock


\bibitem[Covington et~al\mbox{.}(2016)]%
        {covington2016youtube}
\bibfield{author}{\bibinfo{person}{Paul Covington}, \bibinfo{person}{Jay Adams}, {and} \bibinfo{person}{Emre Sargin}.} \bibinfo{year}{2016}\natexlab{}.
\newblock \showarticletitle{Deep Neural Networks for YouTube Recommendations}. In \bibinfo{booktitle}{\emph{Proceedings of the 10th {ACM} Conference on Recommender Systems, Boston, MA, USA, September 15-19, 2016}}, \bibfield{editor}{\bibinfo{person}{Shilad Sen}, \bibinfo{person}{Werner Geyer}, \bibinfo{person}{Jill Freyne}, {and} \bibinfo{person}{Pablo Castells}} (Eds.). \bibinfo{publisher}{{ACM}}, \bibinfo{pages}{191--198}.
\newblock
\href{https://doi.org/10.1145/2959100.2959190}{doi:\nolinkurl{10.1145/2959100.2959190}}


\bibitem[Cui et~al\mbox{.}(2022)]%
        {cui2022m6}
\bibfield{author}{\bibinfo{person}{Zeyu Cui}, \bibinfo{person}{Jianxin Ma}, \bibinfo{person}{Chang Zhou}, \bibinfo{person}{Jingren Zhou}, {and} \bibinfo{person}{Hongxia Yang}.} \bibinfo{year}{2022}\natexlab{}.
\newblock \showarticletitle{M6-rec: Generative pretrained language models are open-ended recommender systems}.
\newblock \bibinfo{journal}{\emph{arXiv preprint arXiv:2205.08084}} (\bibinfo{year}{2022}).
\newblock


\bibitem[Cuturi(2013)]%
        {cuturi2013sinkhorn}
\bibfield{author}{\bibinfo{person}{Marco Cuturi}.} \bibinfo{year}{2013}\natexlab{}.
\newblock \showarticletitle{Sinkhorn Distances: Lightspeed Computation of Optimal Transport}. In \bibinfo{booktitle}{\emph{Advances in Neural Information Processing Systems 26: 27th Annual Conference on Neural Information Processing Systems 2013. Proceedings of a meeting held December 5-8, 2013, Lake Tahoe, Nevada, United States}}, \bibfield{editor}{\bibinfo{person}{Christopher J.~C. Burges}, \bibinfo{person}{L{\'{e}}on Bottou}, \bibinfo{person}{Zoubin Ghahramani}, {and} \bibinfo{person}{Kilian~Q. Weinberger}} (Eds.). \bibinfo{pages}{2292--2300}.
\newblock
\urldef\tempurl%
\url{https://proceedings.neurips.cc/paper/2013/hash/af21d0c97db2e27e13572cbf59eb343d-Abstract.html}
\showURL{%
\tempurl}


\bibitem[Dai et~al\mbox{.}(2025)]%
        {dai2025onepiece}
\bibfield{author}{\bibinfo{person}{Sunhao Dai}, \bibinfo{person}{Jiakai Tang}, \bibinfo{person}{Jiahua Wu}, \bibinfo{person}{Kun Wang}, \bibinfo{person}{Yuxuan Zhu}, \bibinfo{person}{Bingjun Chen}, \bibinfo{person}{Bangyang Hong}, \bibinfo{person}{Yu Zhao}, \bibinfo{person}{Cong Fu}, \bibinfo{person}{Kangle Wu}, {et~al\mbox{.}}} \bibinfo{year}{2025}\natexlab{}.
\newblock \showarticletitle{Onepiece: Bringing context engineering and reasoning to industrial cascade ranking system}.
\newblock \bibinfo{journal}{\emph{arXiv preprint arXiv:2509.18091}} (\bibinfo{year}{2025}).
\newblock


\bibitem[Deng et~al\mbox{.}(2025)]%
        {deng2025onerec}
\bibfield{author}{\bibinfo{person}{Jiaxin Deng}, \bibinfo{person}{Shiyao Wang}, \bibinfo{person}{Kuo Cai}, \bibinfo{person}{Lejian Ren}, \bibinfo{person}{Qigen Hu}, \bibinfo{person}{Weifeng Ding}, \bibinfo{person}{Qiang Luo}, {and} \bibinfo{person}{Guorui Zhou}.} \bibinfo{year}{2025}\natexlab{}.
\newblock \showarticletitle{Onerec: Unifying retrieve and rank with generative recommender and iterative preference alignment}.
\newblock \bibinfo{journal}{\emph{arXiv preprint arXiv:2502.18965}} (\bibinfo{year}{2025}).
\newblock


\bibitem[Feng et~al\mbox{.}(2021)]%
        {feng2021revisit}
\bibfield{author}{\bibinfo{person}{Yufei Feng}, \bibinfo{person}{Yu Gong}, \bibinfo{person}{Fei Sun}, \bibinfo{person}{Junfeng Ge}, {and} \bibinfo{person}{Wenwu Ou}.} \bibinfo{year}{2021}\natexlab{}.
\newblock \showarticletitle{Revisit recommender system in the permutation prospective}.
\newblock \bibinfo{journal}{\emph{arXiv preprint arXiv:2102.12057}} (\bibinfo{year}{2021}).
\newblock


\bibitem[Fu et~al\mbox{.}(2025)]%
        {fu2025forge}
\bibfield{author}{\bibinfo{person}{Kairui Fu}, \bibinfo{person}{Tao Zhang}, \bibinfo{person}{Shuwen Xiao}, \bibinfo{person}{Ziyang Wang}, \bibinfo{person}{Xinming Zhang}, \bibinfo{person}{Chenchi Zhang}, \bibinfo{person}{Yuliang Yan}, \bibinfo{person}{Junjun Zheng}, \bibinfo{person}{Yu Li}, \bibinfo{person}{Zhihong Chen}, \bibinfo{person}{Jian Wu}, \bibinfo{person}{Xiangheng Kong}, \bibinfo{person}{Shengyu Zhang}, \bibinfo{person}{Kun Kuang}, \bibinfo{person}{Yu{-}Gang Jiang}, {and} \bibinfo{person}{Bo Zheng}.} \bibinfo{year}{2025}\natexlab{}.
\newblock \showarticletitle{{FORGE:} Forming Semantic Identifiers for Generative Retrieval in Industrial Datasets}.
\newblock \bibinfo{journal}{\emph{CoRR}}  \bibinfo{volume}{abs/2509.20904} (\bibinfo{year}{2025}).
\newblock
\showeprint[arXiv]{2509.20904}
\href{https://doi.org/10.48550/ARXIV.2509.20904}{doi:\nolinkurl{10.48550/ARXIV.2509.20904}}


\bibitem[Gallagher et~al\mbox{.}(2019)]%
        {gallagher2019cascade}
\bibfield{author}{\bibinfo{person}{Luke Gallagher}, \bibinfo{person}{Ruey{-}Cheng Chen}, \bibinfo{person}{Roi Blanco}, {and} \bibinfo{person}{J.~Shane Culpepper}.} \bibinfo{year}{2019}\natexlab{}.
\newblock \showarticletitle{Joint Optimization of Cascade Ranking Models}. In \bibinfo{booktitle}{\emph{Proceedings of the Twelfth {ACM} International Conference on Web Search and Data Mining, {WSDM} 2019, Melbourne, VIC, Australia, February 11-15, 2019}}, \bibfield{editor}{\bibinfo{person}{J.~Shane Culpepper}, \bibinfo{person}{Alistair Moffat}, \bibinfo{person}{Paul~N. Bennett}, {and} \bibinfo{person}{Kristina Lerman}} (Eds.). \bibinfo{publisher}{{ACM}}, \bibinfo{pages}{15--23}.
\newblock
\href{https://doi.org/10.1145/3289600.3290986}{doi:\nolinkurl{10.1145/3289600.3290986}}


\bibitem[Geng et~al\mbox{.}(2022)]%
        {geng2022recommendation}
\bibfield{author}{\bibinfo{person}{Shijie Geng}, \bibinfo{person}{Shuchang Liu}, \bibinfo{person}{Zuohui Fu}, \bibinfo{person}{Yingqiang Ge}, {and} \bibinfo{person}{Yongfeng Zhang}.} \bibinfo{year}{2022}\natexlab{}.
\newblock \showarticletitle{Recommendation as Language Processing {(RLP):} {A} Unified Pretrain, Personalized Prompt {\&} Predict Paradigm {(P5)}}. In \bibinfo{booktitle}{\emph{RecSys '22: Sixteenth {ACM} Conference on Recommender Systems, Seattle, WA, USA, September 18 - 23, 2022}}, \bibfield{editor}{\bibinfo{person}{Jennifer Golbeck}, \bibinfo{person}{F.~Maxwell Harper}, \bibinfo{person}{Vanessa Murdock}, \bibinfo{person}{Michael~D. Ekstrand}, \bibinfo{person}{Bracha Shapira}, \bibinfo{person}{Justin Basilico}, \bibinfo{person}{Keld~T. Lundgaard}, {and} \bibinfo{person}{Even Oldridge}} (Eds.). \bibinfo{publisher}{{ACM}}, \bibinfo{pages}{299--315}.
\newblock
\href{https://doi.org/10.1145/3523227.3546767}{doi:\nolinkurl{10.1145/3523227.3546767}}


\bibitem[Harper and Konstan(2016)]%
        {harper2015movielens}
\bibfield{author}{\bibinfo{person}{F.~Maxwell Harper} {and} \bibinfo{person}{Joseph~A. Konstan}.} \bibinfo{year}{2016}\natexlab{}.
\newblock \showarticletitle{The MovieLens Datasets: History and Context}.
\newblock \bibinfo{journal}{\emph{{ACM} Trans. Interact. Intell. Syst.}} \bibinfo{volume}{5}, \bibinfo{number}{4} (\bibinfo{year}{2016}), \bibinfo{pages}{19:1--19:19}.
\newblock
\href{https://doi.org/10.1145/2827872}{doi:\nolinkurl{10.1145/2827872}}


\bibitem[Hu et~al\mbox{.}(2025)]%
        {hu2025open}
\bibfield{author}{\bibinfo{person}{Jingcheng Hu}, \bibinfo{person}{Yinmin Zhang}, \bibinfo{person}{Qi Han}, \bibinfo{person}{Daxin Jiang}, \bibinfo{person}{Xiangyu Zhang}, {and} \bibinfo{person}{Heung-Yeung Shum}.} \bibinfo{year}{2025}\natexlab{}.
\newblock \showarticletitle{Open-reasoner-zero: An open source approach to scaling up reinforcement learning on the base model}.
\newblock \bibinfo{journal}{\emph{arXiv preprint arXiv:2503.24290}} (\bibinfo{year}{2025}).
\newblock


\bibitem[Huang et~al\mbox{.}(2025)]%
        {huang2025towards}
\bibfield{author}{\bibinfo{person}{Yanhua Huang}, \bibinfo{person}{Yuqi Chen}, \bibinfo{person}{Xiong Cao}, \bibinfo{person}{Rui Yang}, \bibinfo{person}{Mingliang Qi}, \bibinfo{person}{Yinghao Zhu}, \bibinfo{person}{Qingchang Han}, \bibinfo{person}{Yaowei Liu}, \bibinfo{person}{Zhaoyu Liu}, \bibinfo{person}{Xuefeng Yao}, {et~al\mbox{.}}} \bibinfo{year}{2025}\natexlab{}.
\newblock \showarticletitle{Towards Large-scale Generative Ranking}.
\newblock \bibinfo{journal}{\emph{arXiv preprint arXiv:2505.04180}} (\bibinfo{year}{2025}).
\newblock


\bibitem[Huzhang et~al\mbox{.}(2023)]%
        {huzhang2021aliexpress}
\bibfield{author}{\bibinfo{person}{Guangda Huzhang}, \bibinfo{person}{Zhen{-}Jia Pang}, \bibinfo{person}{Yongqing Gao}, \bibinfo{person}{Yawen Liu}, \bibinfo{person}{Weijie Shen}, \bibinfo{person}{Wen{-}Ji Zhou}, \bibinfo{person}{Qianying Lin}, \bibinfo{person}{Qing Da}, \bibinfo{person}{Anxiang Zeng}, \bibinfo{person}{Han Yu}, \bibinfo{person}{Yang Yu}, {and} \bibinfo{person}{Zhi{-}Hua Zhou}.} \bibinfo{year}{2023}\natexlab{}.
\newblock \showarticletitle{AliExpress Learning-to-Rank: Maximizing Online Model Performance Without Going Online}.
\newblock \bibinfo{journal}{\emph{{IEEE} Trans. Knowl. Data Eng.}} \bibinfo{volume}{35}, \bibinfo{number}{2} (\bibinfo{year}{2023}), \bibinfo{pages}{1214--1226}.
\newblock
\href{https://doi.org/10.1109/TKDE.2021.3098898}{doi:\nolinkurl{10.1109/TKDE.2021.3098898}}


\bibitem[Ju et~al\mbox{.}(2025)]%
        {ju2025generative}
\bibfield{author}{\bibinfo{person}{Clark~Mingxuan Ju}, \bibinfo{person}{Liam Collins}, \bibinfo{person}{Leonardo Neves}, \bibinfo{person}{Bhuvesh Kumar}, \bibinfo{person}{Louis~Yufeng Wang}, \bibinfo{person}{Tong Zhao}, {and} \bibinfo{person}{Neil Shah}.} \bibinfo{year}{2025}\natexlab{}.
\newblock \showarticletitle{Generative Recommendation with Semantic IDs: {A} Practitioner's Handbook}. In \bibinfo{booktitle}{\emph{Proceedings of the 34th {ACM} International Conference on Information and Knowledge Management, {CIKM} 2025, Seoul, Republic of Korea, November 10-14, 2025}}, \bibfield{editor}{\bibinfo{person}{Meeyoung Cha}, \bibinfo{person}{Chanyoung Park}, \bibinfo{person}{Noseong Park}, \bibinfo{person}{Carl Yang}, \bibinfo{person}{Senjuti~Basu Roy}, \bibinfo{person}{Jessie Li}, \bibinfo{person}{Jaap Kamps}, \bibinfo{person}{Kijung Shin}, \bibinfo{person}{Bryan Hooi}, {and} \bibinfo{person}{Lifang He}} (Eds.). \bibinfo{publisher}{{ACM}}, \bibinfo{pages}{6420--6425}.
\newblock
\href{https://doi.org/10.1145/3746252.3761612}{doi:\nolinkurl{10.1145/3746252.3761612}}


\bibitem[Kong et~al\mbox{.}(2025)]%
        {DBLP:journals/corr/abs-2510-24431}
\bibfield{author}{\bibinfo{person}{Xiaoyu Kong}, \bibinfo{person}{Leheng Sheng}, \bibinfo{person}{Junfei Tan}, \bibinfo{person}{Yuxin Chen}, \bibinfo{person}{Jiancan Wu}, \bibinfo{person}{An Zhang}, \bibinfo{person}{Xiang Wang}, {and} \bibinfo{person}{Xiangnan He}.} \bibinfo{year}{2025}\natexlab{}.
\newblock \showarticletitle{MiniOneRec: An Open-Source Framework for Scaling Generative Recommendation}.
\newblock \bibinfo{journal}{\emph{CoRR}}  \bibinfo{volume}{abs/2510.24431} (\bibinfo{year}{2025}).
\newblock
\showeprint[arXiv]{2510.24431}
\href{https://doi.org/10.48550/ARXIV.2510.24431}{doi:\nolinkurl{10.48550/ARXIV.2510.24431}}


\bibitem[Li et~al\mbox{.}(2025)]%
        {li2025survey}
\bibfield{author}{\bibinfo{person}{Xiaopeng Li}, \bibinfo{person}{Bo Chen}, \bibinfo{person}{Junda She}, \bibinfo{person}{Shiteng Cao}, \bibinfo{person}{You Wang}, \bibinfo{person}{Qinlin Jia}, \bibinfo{person}{Haiying He}, \bibinfo{person}{Zheli Zhou}, \bibinfo{person}{Zhao Liu}, \bibinfo{person}{Ji Liu}, {et~al\mbox{.}}} \bibinfo{year}{2025}\natexlab{}.
\newblock \showarticletitle{A Survey of Generative Recommendation from a Tri-Decoupled Perspective: Tokenization, Architecture, and Optimization}.
\newblock  (\bibinfo{year}{2025}).
\newblock


\bibitem[Liao et~al\mbox{.}(2024)]%
        {DBLP:conf/sigir/LiaoL0WYW024}
\bibfield{author}{\bibinfo{person}{Jiayi Liao}, \bibinfo{person}{Sihang Li}, \bibinfo{person}{Zhengyi Yang}, \bibinfo{person}{Jiancan Wu}, \bibinfo{person}{Yancheng Yuan}, \bibinfo{person}{Xiang Wang}, {and} \bibinfo{person}{Xiangnan He}.} \bibinfo{year}{2024}\natexlab{}.
\newblock \showarticletitle{LLaRA: Large Language-Recommendation Assistant}. In \bibinfo{booktitle}{\emph{Proceedings of the 47th International {ACM} {SIGIR} Conference on Research and Development in Information Retrieval, {SIGIR} 2024, Washington DC, USA, July 14-18, 2024}}, \bibfield{editor}{\bibinfo{person}{Grace~Hui Yang}, \bibinfo{person}{Hongning Wang}, \bibinfo{person}{Sam Han}, \bibinfo{person}{Claudia Hauff}, \bibinfo{person}{Guido Zuccon}, {and} \bibinfo{person}{Yi~Zhang}} (Eds.). \bibinfo{publisher}{{ACM}}, \bibinfo{pages}{1785--1795}.
\newblock
\href{https://doi.org/10.1145/3626772.3657690}{doi:\nolinkurl{10.1145/3626772.3657690}}


\bibitem[Lin et~al\mbox{.}(2024)]%
        {lin2024dcdr}
\bibfield{author}{\bibinfo{person}{Xiao Lin}, \bibinfo{person}{Xiaokai Chen}, \bibinfo{person}{Chenyang Wang}, \bibinfo{person}{Hantao Shu}, \bibinfo{person}{Linfeng Song}, \bibinfo{person}{Biao Li}, {and} \bibinfo{person}{Peng Jiang}.} \bibinfo{year}{2024}\natexlab{}.
\newblock \showarticletitle{Discrete Conditional Diffusion for Reranking in Recommendation}. In \bibinfo{booktitle}{\emph{Companion Proceedings of the {ACM} on Web Conference 2024, {WWW} 2024, Singapore, Singapore, May 13-17, 2024}}, \bibfield{editor}{\bibinfo{person}{Tat{-}Seng Chua}, \bibinfo{person}{Chong{-}Wah Ngo}, \bibinfo{person}{Roy~Ka{-}Wei Lee}, \bibinfo{person}{Ravi Kumar}, {and} \bibinfo{person}{Hady~W. Lauw}} (Eds.). \bibinfo{publisher}{{ACM}}, \bibinfo{pages}{161--169}.
\newblock
\href{https://doi.org/10.1145/3589335.3648313}{doi:\nolinkurl{10.1145/3589335.3648313}}


\bibitem[Lin et~al\mbox{.}(2025)]%
        {lin2025gref}
\bibfield{author}{\bibinfo{person}{Zhijie Lin}, \bibinfo{person}{Zhuofeng Li}, \bibinfo{person}{Chenglei Dai}, \bibinfo{person}{Wentian Bao}, \bibinfo{person}{Shuai Lin}, \bibinfo{person}{Enyun Yu}, \bibinfo{person}{Haoxiang Zhang}, {and} \bibinfo{person}{Liang Zhao}.} \bibinfo{year}{2025}\natexlab{}.
\newblock \showarticletitle{GReF: {A} Unified Generative Framework for Efficient Reranking via Ordered Multi-token Prediction}. In \bibinfo{booktitle}{\emph{Proceedings of the 34th {ACM} International Conference on Information and Knowledge Management, {CIKM} 2025, Seoul, Republic of Korea, November 10-14, 2025}}, \bibfield{editor}{\bibinfo{person}{Meeyoung Cha}, \bibinfo{person}{Chanyoung Park}, \bibinfo{person}{Noseong Park}, \bibinfo{person}{Carl Yang}, \bibinfo{person}{Senjuti~Basu Roy}, \bibinfo{person}{Jessie Li}, \bibinfo{person}{Jaap Kamps}, \bibinfo{person}{Kijung Shin}, \bibinfo{person}{Bryan Hooi}, {and} \bibinfo{person}{Lifang He}} (Eds.). \bibinfo{publisher}{{ACM}}, \bibinfo{pages}{5879--5887}.
\newblock
\href{https://doi.org/10.1145/3746252.3761540}{doi:\nolinkurl{10.1145/3746252.3761540}}


\bibitem[Liu et~al\mbox{.}(2025b)]%
        {liu2025recflow}
\bibfield{author}{\bibinfo{person}{Qi Liu}, \bibinfo{person}{Kai Zheng}, \bibinfo{person}{Rui Huang}, \bibinfo{person}{Wuchao Li}, \bibinfo{person}{Kuo Cai}, \bibinfo{person}{Yuan Chai}, \bibinfo{person}{Yanan Niu}, \bibinfo{person}{Yiqun Hui}, \bibinfo{person}{Bing Han}, \bibinfo{person}{Na Mou}, \bibinfo{person}{Hongning Wang}, \bibinfo{person}{Wentian Bao}, \bibinfo{person}{Yunen Yu}, \bibinfo{person}{Guorui Zhou}, \bibinfo{person}{Han Li}, \bibinfo{person}{Yang Song}, \bibinfo{person}{Defu Lian}, {and} \bibinfo{person}{Kun Gai}.} \bibinfo{year}{2025}\natexlab{b}.
\newblock \showarticletitle{RecFlow: An Industrial Full Flow Recommendation Dataset}. In \bibinfo{booktitle}{\emph{The Thirteenth International Conference on Learning Representations, {ICLR} 2025, Singapore, April 24-28, 2025}}. \bibinfo{publisher}{OpenReview.net}.
\newblock
\urldef\tempurl%
\url{https://openreview.net/forum?id=vVHc8bGRns}
\showURL{%
\tempurl}


\bibitem[Liu et~al\mbox{.}(2023)]%
        {liu2023gfn4list}
\bibfield{author}{\bibinfo{person}{Shuchang Liu}, \bibinfo{person}{Qingpeng Cai}, \bibinfo{person}{Zhankui He}, \bibinfo{person}{Bowen Sun}, \bibinfo{person}{Julian~J. McAuley}, \bibinfo{person}{Dong Zheng}, \bibinfo{person}{Peng Jiang}, {and} \bibinfo{person}{Kun Gai}.} \bibinfo{year}{2023}\natexlab{}.
\newblock \showarticletitle{Generative Flow Network for Listwise Recommendation}. In \bibinfo{booktitle}{\emph{Proceedings of the 29th {ACM} {SIGKDD} Conference on Knowledge Discovery and Data Mining, {KDD} 2023, Long Beach, CA, USA, August 6-10, 2023}}, \bibfield{editor}{\bibinfo{person}{Ambuj~K. Singh}, \bibinfo{person}{Yizhou Sun}, \bibinfo{person}{Leman Akoglu}, \bibinfo{person}{Dimitrios Gunopulos}, \bibinfo{person}{Xifeng Yan}, \bibinfo{person}{Ravi Kumar}, \bibinfo{person}{Fatma Ozcan}, {and} \bibinfo{person}{Jieping Ye}} (Eds.). \bibinfo{publisher}{{ACM}}, \bibinfo{pages}{1524--1534}.
\newblock
\href{https://doi.org/10.1145/3580305.3599364}{doi:\nolinkurl{10.1145/3580305.3599364}}


\bibitem[Liu et~al\mbox{.}(2022)]%
        {liu2022rerankingsurvey}
\bibfield{author}{\bibinfo{person}{Weiwen Liu}, \bibinfo{person}{Yunjia Xi}, \bibinfo{person}{Jiarui Qin}, \bibinfo{person}{Fei Sun}, \bibinfo{person}{Bo Chen}, \bibinfo{person}{Weinan Zhang}, \bibinfo{person}{Rui Zhang}, {and} \bibinfo{person}{Ruiming Tang}.} \bibinfo{year}{2022}\natexlab{}.
\newblock \showarticletitle{Neural Re-ranking in Multi-stage Recommender Systems: {A} Review}. In \bibinfo{booktitle}{\emph{Proceedings of the Thirty-First International Joint Conference on Artificial Intelligence, {IJCAI} 2022, Vienna, Austria, 23-29 July 2022}}, \bibfield{editor}{\bibinfo{person}{Luc~De Raedt}} (Ed.). \bibinfo{publisher}{ijcai.org}, \bibinfo{pages}{5512--5520}.
\newblock
\href{https://doi.org/10.24963/IJCAI.2022/771}{doi:\nolinkurl{10.24963/IJCAI.2022/771}}


\bibitem[Liu et~al\mbox{.}(2025a)]%
        {liu2025understanding}
\bibfield{author}{\bibinfo{person}{Zichen Liu}, \bibinfo{person}{Changyu Chen}, \bibinfo{person}{Wenjun Li}, \bibinfo{person}{Penghui Qi}, \bibinfo{person}{Tianyu Pang}, \bibinfo{person}{Chao Du}, \bibinfo{person}{Wee~Sun Lee}, {and} \bibinfo{person}{Min Lin}.} \bibinfo{year}{2025}\natexlab{a}.
\newblock \showarticletitle{Understanding r1-zero-like training: A critical perspective}.
\newblock \bibinfo{journal}{\emph{arXiv preprint arXiv:2503.20783}} (\bibinfo{year}{2025}).
\newblock


\bibitem[Loshchilov and Hutter(2019)]%
        {loshchilov2017decoupled}
\bibfield{author}{\bibinfo{person}{Ilya Loshchilov} {and} \bibinfo{person}{Frank Hutter}.} \bibinfo{year}{2019}\natexlab{}.
\newblock \showarticletitle{Decoupled Weight Decay Regularization}. In \bibinfo{booktitle}{\emph{7th International Conference on Learning Representations, {ICLR} 2019, New Orleans, LA, USA, May 6-9, 2019}}. \bibinfo{publisher}{OpenReview.net}.
\newblock
\urldef\tempurl%
\url{https://openreview.net/forum?id=Bkg6RiCqY7}
\showURL{%
\tempurl}


\bibitem[Mao et~al\mbox{.}(2025)]%
        {mao2025denoising}
\bibfield{author}{\bibinfo{person}{Wenyu Mao}, \bibinfo{person}{Shuchang Liu}, \bibinfo{person}{Hailan Yang}, \bibinfo{person}{Xiaobei Wang}, \bibinfo{person}{Xiaoyu Yang}, \bibinfo{person}{Xu Gao}, \bibinfo{person}{Xiang Li}, \bibinfo{person}{Lantao Hu}, \bibinfo{person}{Han Li}, \bibinfo{person}{Kun Gai}, {et~al\mbox{.}}} \bibinfo{year}{2025}\natexlab{}.
\newblock \showarticletitle{Denoising Neural Reranker for Recommender Systems}.
\newblock \bibinfo{journal}{\emph{arXiv preprint arXiv:2509.18736}} (\bibinfo{year}{2025}).
\newblock


\bibitem[McAuley et~al\mbox{.}(2015)]%
        {mcauley2015image}
\bibfield{author}{\bibinfo{person}{Julian~J. McAuley}, \bibinfo{person}{Christopher Targett}, \bibinfo{person}{Qinfeng Shi}, {and} \bibinfo{person}{Anton van~den Hengel}.} \bibinfo{year}{2015}\natexlab{}.
\newblock \showarticletitle{Image-Based Recommendations on Styles and Substitutes}. In \bibinfo{booktitle}{\emph{Proceedings of the 38th International {ACM} {SIGIR} Conference on Research and Development in Information Retrieval, Santiago, Chile, August 9-13, 2015}}, \bibfield{editor}{\bibinfo{person}{Ricardo Baeza{-}Yates}, \bibinfo{person}{Mounia Lalmas}, \bibinfo{person}{Alistair Moffat}, {and} \bibinfo{person}{Berthier~A. Ribeiro{-}Neto}} (Eds.). \bibinfo{publisher}{{ACM}}, \bibinfo{pages}{43--52}.
\newblock
\href{https://doi.org/10.1145/2766462.2767755}{doi:\nolinkurl{10.1145/2766462.2767755}}


\bibitem[Meng et~al\mbox{.}(2025)]%
        {meng2025generative}
\bibfield{author}{\bibinfo{person}{Yue Meng}, \bibinfo{person}{Cheng Guo}, \bibinfo{person}{Yi Cao}, \bibinfo{person}{Tong Liu}, {and} \bibinfo{person}{Bo Zheng}.} \bibinfo{year}{2025}\natexlab{}.
\newblock \showarticletitle{A Generative Re-ranking Model for List-level Multi-objective Optimization at Taobao}. In \bibinfo{booktitle}{\emph{Proceedings of the 48th International {ACM} {SIGIR} Conference on Research and Development in Information Retrieval, {SIGIR} 2025, Padua, Italy, July 13-18, 2025}}, \bibfield{editor}{\bibinfo{person}{Nicola Ferro}, \bibinfo{person}{Maria Maistro}, \bibinfo{person}{Gabriella Pasi}, \bibinfo{person}{Omar Alonso}, \bibinfo{person}{Andrew Trotman}, {and} \bibinfo{person}{Suzan Verberne}} (Eds.). \bibinfo{publisher}{{ACM}}, \bibinfo{pages}{4213--4218}.
\newblock
\href{https://doi.org/10.1145/3726302.3731935}{doi:\nolinkurl{10.1145/3726302.3731935}}


\bibitem[Ni et~al\mbox{.}(2022)]%
        {ni2022sentence}
\bibfield{author}{\bibinfo{person}{Jianmo Ni}, \bibinfo{person}{Gustavo~Hern{\'{a}}ndez {\'{A}}brego}, \bibinfo{person}{Noah Constant}, \bibinfo{person}{Ji Ma}, \bibinfo{person}{Keith~B. Hall}, \bibinfo{person}{Daniel Cer}, {and} \bibinfo{person}{Yinfei Yang}.} \bibinfo{year}{2022}\natexlab{}.
\newblock \showarticletitle{Sentence-T5: Scalable Sentence Encoders from Pre-trained Text-to-Text Models}. In \bibinfo{booktitle}{\emph{Findings of the Association for Computational Linguistics: {ACL} 2022, Dublin, Ireland, May 22-27, 2022}}, \bibfield{editor}{\bibinfo{person}{Smaranda Muresan}, \bibinfo{person}{Preslav Nakov}, {and} \bibinfo{person}{Aline Villavicencio}} (Eds.). \bibinfo{publisher}{Association for Computational Linguistics}, \bibinfo{pages}{1864--1874}.
\newblock
\href{https://doi.org/10.18653/V1/2022.FINDINGS-ACL.146}{doi:\nolinkurl{10.18653/V1/2022.FINDINGS-ACL.146}}


\bibitem[Pei et~al\mbox{.}(2019)]%
        {pei2019prm}
\bibfield{author}{\bibinfo{person}{Changhua Pei}, \bibinfo{person}{Yi Zhang}, \bibinfo{person}{Yongfeng Zhang}, \bibinfo{person}{Fei Sun}, \bibinfo{person}{Xiao Lin}, \bibinfo{person}{Hanxiao Sun}, \bibinfo{person}{Jian Wu}, \bibinfo{person}{Peng Jiang}, \bibinfo{person}{Junfeng Ge}, \bibinfo{person}{Wenwu Ou}, {and} \bibinfo{person}{Dan Pei}.} \bibinfo{year}{2019}\natexlab{}.
\newblock \showarticletitle{Personalized re-ranking for recommendation}. In \bibinfo{booktitle}{\emph{Proceedings of the 13th {ACM} Conference on Recommender Systems, RecSys 2019, Copenhagen, Denmark, September 16-20, 2019}}, \bibfield{editor}{\bibinfo{person}{Toine Bogers}, \bibinfo{person}{Alan Said}, \bibinfo{person}{Peter Brusilovsky}, {and} \bibinfo{person}{Domonkos Tikk}} (Eds.). \bibinfo{publisher}{{ACM}}, \bibinfo{pages}{3--11}.
\newblock
\href{https://doi.org/10.1145/3298689.3347000}{doi:\nolinkurl{10.1145/3298689.3347000}}


\bibitem[Raffel et~al\mbox{.}(2020)]%
        {2020t5}
\bibfield{author}{\bibinfo{person}{Colin Raffel}, \bibinfo{person}{Noam Shazeer}, \bibinfo{person}{Adam Roberts}, \bibinfo{person}{Katherine Lee}, \bibinfo{person}{Sharan Narang}, \bibinfo{person}{Michael Matena}, \bibinfo{person}{Yanqi Zhou}, \bibinfo{person}{Wei Li}, {and} \bibinfo{person}{Peter~J. Liu}.} \bibinfo{year}{2020}\natexlab{}.
\newblock \showarticletitle{Exploring the Limits of Transfer Learning with a Unified Text-to-Text Transformer}.
\newblock \bibinfo{journal}{\emph{J. Mach. Learn. Res.}}  \bibinfo{volume}{21} (\bibinfo{year}{2020}), \bibinfo{pages}{140:1--140:67}.
\newblock
\urldef\tempurl%
\url{https://jmlr.org/papers/v21/20-074.html}
\showURL{%
\tempurl}


\bibitem[Rajput et~al\mbox{.}(2023)]%
        {rajput2023recommender}
\bibfield{author}{\bibinfo{person}{Shashank Rajput}, \bibinfo{person}{Nikhil Mehta}, \bibinfo{person}{Anima Singh}, \bibinfo{person}{Raghunandan~Hulikal Keshavan}, \bibinfo{person}{Trung Vu}, \bibinfo{person}{Lukasz Heldt}, \bibinfo{person}{Lichan Hong}, \bibinfo{person}{Yi Tay}, \bibinfo{person}{Vinh~Q. Tran}, \bibinfo{person}{Jonah Samost}, \bibinfo{person}{Maciej Kula}, \bibinfo{person}{Ed~H. Chi}, {and} \bibinfo{person}{Mahesh Sathiamoorthy}.} \bibinfo{year}{2023}\natexlab{}.
\newblock \showarticletitle{Recommender Systems with Generative Retrieval}. In \bibinfo{booktitle}{\emph{Advances in Neural Information Processing Systems 36: Annual Conference on Neural Information Processing Systems 2023, NeurIPS 2023, New Orleans, LA, USA, December 10 - 16, 2023}}, \bibfield{editor}{\bibinfo{person}{Alice Oh}, \bibinfo{person}{Tristan Naumann}, \bibinfo{person}{Amir Globerson}, \bibinfo{person}{Kate Saenko}, \bibinfo{person}{Moritz Hardt}, {and} \bibinfo{person}{Sergey Levine}} (Eds.).
\newblock
\urldef\tempurl%
\url{http://papers.nips.cc/paper\_files/paper/2023/hash/20dcab0f14046a5c6b02b61da9f13229-Abstract-Conference.html}
\showURL{%
\tempurl}


\bibitem[Ren et~al\mbox{.}(2024)]%
        {ren2024nar4rec}
\bibfield{author}{\bibinfo{person}{Yuxin Ren}, \bibinfo{person}{Qiya Yang}, \bibinfo{person}{Yichun Wu}, \bibinfo{person}{Wei Xu}, \bibinfo{person}{Yalong Wang}, {and} \bibinfo{person}{Zhiqiang Zhang}.} \bibinfo{year}{2024}\natexlab{}.
\newblock \showarticletitle{Non-autoregressive Generative Models for Reranking Recommendation}. In \bibinfo{booktitle}{\emph{Proceedings of the 30th {ACM} {SIGKDD} Conference on Knowledge Discovery and Data Mining, {KDD} 2024, Barcelona, Spain, August 25-29, 2024}}, \bibfield{editor}{\bibinfo{person}{Ricardo Baeza{-}Yates} {and} \bibinfo{person}{Francesco Bonchi}} (Eds.). \bibinfo{publisher}{{ACM}}, \bibinfo{pages}{5625--5634}.
\newblock
\href{https://doi.org/10.1145/3637528.3671645}{doi:\nolinkurl{10.1145/3637528.3671645}}


\bibitem[Shao et~al\mbox{.}(2024)]%
        {shao2024deepseekmath}
\bibfield{author}{\bibinfo{person}{Zhihong Shao}, \bibinfo{person}{Peiyi Wang}, \bibinfo{person}{Qihao Zhu}, \bibinfo{person}{Runxin Xu}, \bibinfo{person}{Junxiao Song}, \bibinfo{person}{Xiao Bi}, \bibinfo{person}{Haowei Zhang}, \bibinfo{person}{Mingchuan Zhang}, \bibinfo{person}{YK Li}, \bibinfo{person}{Yang Wu}, {et~al\mbox{.}}} \bibinfo{year}{2024}\natexlab{}.
\newblock \showarticletitle{Deepseekmath: Pushing the limits of mathematical reasoning in open language models}.
\newblock \bibinfo{journal}{\emph{arXiv preprint arXiv:2402.03300}} (\bibinfo{year}{2024}).
\newblock


\bibitem[Shi et~al\mbox{.}(2023)]%
        {shi2023pier}
\bibfield{author}{\bibinfo{person}{Xiaowen Shi}, \bibinfo{person}{Fan Yang}, \bibinfo{person}{Ze Wang}, \bibinfo{person}{Xiaoxu Wu}, \bibinfo{person}{Muzhi Guan}, \bibinfo{person}{Guogang Liao}, \bibinfo{person}{Yongkang Wang}, \bibinfo{person}{Xingxing Wang}, {and} \bibinfo{person}{Dong Wang}.} \bibinfo{year}{2023}\natexlab{}.
\newblock \showarticletitle{{PIER:} Permutation-Level Interest-Based End-to-End Re-ranking Framework in E-commerce}. In \bibinfo{booktitle}{\emph{Proceedings of the 29th {ACM} {SIGKDD} Conference on Knowledge Discovery and Data Mining, {KDD} 2023, Long Beach, CA, USA, August 6-10, 2023}}, \bibfield{editor}{\bibinfo{person}{Ambuj~K. Singh}, \bibinfo{person}{Yizhou Sun}, \bibinfo{person}{Leman Akoglu}, \bibinfo{person}{Dimitrios Gunopulos}, \bibinfo{person}{Xifeng Yan}, \bibinfo{person}{Ravi Kumar}, \bibinfo{person}{Fatma Ozcan}, {and} \bibinfo{person}{Jieping Ye}} (Eds.). \bibinfo{publisher}{{ACM}}, \bibinfo{pages}{4823--4831}.
\newblock
\href{https://doi.org/10.1145/3580305.3599886}{doi:\nolinkurl{10.1145/3580305.3599886}}


\bibitem[Wang et~al\mbox{.}(2025)]%
        {wang2025nlgr}
\bibfield{author}{\bibinfo{person}{Shuli Wang}, \bibinfo{person}{Xue Wei}, \bibinfo{person}{Senjie Kou}, \bibinfo{person}{Chi Wang}, \bibinfo{person}{Wenshuai Chen}, \bibinfo{person}{Qi Tang}, \bibinfo{person}{Yinhua Zhu}, \bibinfo{person}{Xiong Xiao}, {and} \bibinfo{person}{Xingxing Wang}.} \bibinfo{year}{2025}\natexlab{}.
\newblock \showarticletitle{{NLGR:} Utilizing Neighbor Lists for Generative Rerank in Personalized Recommendation Systems}. In \bibinfo{booktitle}{\emph{Companion Proceedings of the {ACM} on Web Conference 2025, {WWW} 2025, Sydney, NSW, Australia, 28 April 2025 - 2 May 2025}}, \bibfield{editor}{\bibinfo{person}{Guodong Long}, \bibinfo{person}{Michale Blumestein}, \bibinfo{person}{Yi~Chang}, \bibinfo{person}{Liane Lewin{-}Eytan}, \bibinfo{person}{Zi~Helen Huang}, {and} \bibinfo{person}{Elad Yom{-}Tov}} (Eds.). \bibinfo{publisher}{{ACM}}, \bibinfo{pages}{530--537}.
\newblock
\href{https://doi.org/10.1145/3701716.3715251}{doi:\nolinkurl{10.1145/3701716.3715251}}


\bibitem[Wang et~al\mbox{.}(2024)]%
        {wang2024learnable}
\bibfield{author}{\bibinfo{person}{Wenjie Wang}, \bibinfo{person}{Honghui Bao}, \bibinfo{person}{Xinyu Lin}, \bibinfo{person}{Jizhi Zhang}, \bibinfo{person}{Yongqi Li}, \bibinfo{person}{Fuli Feng}, \bibinfo{person}{See{-}Kiong Ng}, {and} \bibinfo{person}{Tat{-}Seng Chua}.} \bibinfo{year}{2024}\natexlab{}.
\newblock \showarticletitle{Learnable Item Tokenization for Generative Recommendation}. In \bibinfo{booktitle}{\emph{Proceedings of the 33rd {ACM} International Conference on Information and Knowledge Management, {CIKM} 2024, Boise, ID, USA, October 21-25, 2024}}, \bibfield{editor}{\bibinfo{person}{Edoardo Serra} {and} \bibinfo{person}{Francesca Spezzano}} (Eds.). \bibinfo{publisher}{{ACM}}, \bibinfo{pages}{2400--2409}.
\newblock
\href{https://doi.org/10.1145/3627673.3679569}{doi:\nolinkurl{10.1145/3627673.3679569}}


\bibitem[Xi et~al\mbox{.}(2022)]%
        {xi2022multi}
\bibfield{author}{\bibinfo{person}{Yunjia Xi}, \bibinfo{person}{Weiwen Liu}, \bibinfo{person}{Jieming Zhu}, \bibinfo{person}{Xilong Zhao}, \bibinfo{person}{Xinyi Dai}, \bibinfo{person}{Ruiming Tang}, \bibinfo{person}{Weinan Zhang}, \bibinfo{person}{Rui Zhang}, {and} \bibinfo{person}{Yong Yu}.} \bibinfo{year}{2022}\natexlab{}.
\newblock \showarticletitle{Multi-Level Interaction Reranking with User Behavior History}. In \bibinfo{booktitle}{\emph{{SIGIR} '22: The 45th International {ACM} {SIGIR} Conference on Research and Development in Information Retrieval, Madrid, Spain, July 11 - 15, 2022}}, \bibfield{editor}{\bibinfo{person}{Enrique Amig{\'{o}}}, \bibinfo{person}{Pablo Castells}, \bibinfo{person}{Julio Gonzalo}, \bibinfo{person}{Ben Carterette}, \bibinfo{person}{J.~Shane Culpepper}, {and} \bibinfo{person}{Gabriella Kazai}} (Eds.). \bibinfo{publisher}{{ACM}}, \bibinfo{pages}{1336--1346}.
\newblock
\href{https://doi.org/10.1145/3477495.3532026}{doi:\nolinkurl{10.1145/3477495.3532026}}


\bibitem[Yang et~al\mbox{.}(2025)]%
        {yang2025mgeclig}
\bibfield{author}{\bibinfo{person}{Hailan Yang}, \bibinfo{person}{Zhenyu Qi}, \bibinfo{person}{Shuchang Liu}, \bibinfo{person}{Xiaoyu Yang}, \bibinfo{person}{Xiaobei Wang}, \bibinfo{person}{Xiang Li}, \bibinfo{person}{Lantao Hu}, \bibinfo{person}{Han Li}, {and} \bibinfo{person}{Kun Gai}.} \bibinfo{year}{2025}\natexlab{}.
\newblock \showarticletitle{Comprehensive List Generation for Multi-Generator Reranking}. In \bibinfo{booktitle}{\emph{Proceedings of the 48th International {ACM} {SIGIR} Conference on Research and Development in Information Retrieval, {SIGIR} 2025, Padua, Italy, July 13-18, 2025}}, \bibfield{editor}{\bibinfo{person}{Nicola Ferro}, \bibinfo{person}{Maria Maistro}, \bibinfo{person}{Gabriella Pasi}, \bibinfo{person}{Omar Alonso}, \bibinfo{person}{Andrew Trotman}, {and} \bibinfo{person}{Suzan Verberne}} (Eds.). \bibinfo{publisher}{{ACM}}, \bibinfo{pages}{2298--2308}.
\newblock
\href{https://doi.org/10.1145/3726302.3729933}{doi:\nolinkurl{10.1145/3726302.3729933}}


\bibitem[Yu et~al\mbox{.}(2025)]%
        {yu2025dapo}
\bibfield{author}{\bibinfo{person}{Qiying Yu}, \bibinfo{person}{Zheng Zhang}, \bibinfo{person}{Ruofei Zhu}, \bibinfo{person}{Yufeng Yuan}, \bibinfo{person}{Xiaochen Zuo}, \bibinfo{person}{Yu Yue}, \bibinfo{person}{Weinan Dai}, \bibinfo{person}{Tiantian Fan}, \bibinfo{person}{Gaohong Liu}, \bibinfo{person}{Lingjun Liu}, {et~al\mbox{.}}} \bibinfo{year}{2025}\natexlab{}.
\newblock \showarticletitle{Dapo: An open-source llm reinforcement learning system at scale}.
\newblock \bibinfo{journal}{\emph{arXiv preprint arXiv:2503.14476}} (\bibinfo{year}{2025}).
\newblock


\bibitem[Zhai et~al\mbox{.}(2024)]%
        {DBLP:conf/icml/ZhaiLLWLCGGGHLS24}
\bibfield{author}{\bibinfo{person}{Jiaqi Zhai}, \bibinfo{person}{Lucy Liao}, \bibinfo{person}{Xing Liu}, \bibinfo{person}{Yueming Wang}, \bibinfo{person}{Rui Li}, \bibinfo{person}{Xuan Cao}, \bibinfo{person}{Leon Gao}, \bibinfo{person}{Zhaojie Gong}, \bibinfo{person}{Fangda Gu}, \bibinfo{person}{Jiayuan He}, \bibinfo{person}{Yinghai Lu}, {and} \bibinfo{person}{Yu Shi}.} \bibinfo{year}{2024}\natexlab{}.
\newblock \showarticletitle{Actions Speak Louder than Words: Trillion-Parameter Sequential Transducers for Generative Recommendations}. In \bibinfo{booktitle}{\emph{Forty-first International Conference on Machine Learning, {ICML} 2024, Vienna, Austria, July 21-27, 2024}}. \bibinfo{publisher}{OpenReview.net}.
\newblock
\urldef\tempurl%
\url{https://openreview.net/forum?id=xye7iNsgXn}
\showURL{%
\tempurl}


\bibitem[Zhang et~al\mbox{.}(2025b)]%
        {zhang2025goalrank}
\bibfield{author}{\bibinfo{person}{Kaike Zhang}, \bibinfo{person}{Xiaobei Wang}, \bibinfo{person}{Shuchang Liu}, \bibinfo{person}{Hailan Yang}, \bibinfo{person}{Xiang Li}, \bibinfo{person}{Lantao Hu}, \bibinfo{person}{Han Li}, \bibinfo{person}{Qi Cao}, \bibinfo{person}{Fei Sun}, {and} \bibinfo{person}{Kun Gai}.} \bibinfo{year}{2025}\natexlab{b}.
\newblock \showarticletitle{GoalRank: Group-Relative Optimization for a Large Ranking Model}.
\newblock \bibinfo{journal}{\emph{CoRR}}  \bibinfo{volume}{abs/2509.22046} (\bibinfo{year}{2025}).
\newblock
\showeprint[arXiv]{2509.22046}
\href{https://doi.org/10.48550/ARXIV.2509.22046}{doi:\nolinkurl{10.48550/ARXIV.2509.22046}}


\bibitem[Zhang et~al\mbox{.}(2025a)]%
        {qwen3embedding}
\bibfield{author}{\bibinfo{person}{Yanzhao Zhang}, \bibinfo{person}{Mingxin Li}, \bibinfo{person}{Dingkun Long}, \bibinfo{person}{Xin Zhang}, \bibinfo{person}{Huan Lin}, \bibinfo{person}{Baosong Yang}, \bibinfo{person}{Pengjun Xie}, \bibinfo{person}{An Yang}, \bibinfo{person}{Dayiheng Liu}, \bibinfo{person}{Junyang Lin}, \bibinfo{person}{Fei Huang}, {and} \bibinfo{person}{Jingren Zhou}.} \bibinfo{year}{2025}\natexlab{a}.
\newblock \showarticletitle{Qwen3 Embedding: Advancing Text Embedding and Reranking Through Foundation Models}.
\newblock \bibinfo{journal}{\emph{arXiv preprint arXiv:2506.05176}} (\bibinfo{year}{2025}).
\newblock


\bibitem[Zhao et~al\mbox{.}(2023)]%
        {zhao2023kuaisim}
\bibfield{author}{\bibinfo{person}{Kesen Zhao}, \bibinfo{person}{Shuchang Liu}, \bibinfo{person}{Qingpeng Cai}, \bibinfo{person}{Xiangyu Zhao}, \bibinfo{person}{Ziru Liu}, \bibinfo{person}{Dong Zheng}, \bibinfo{person}{Peng Jiang}, {and} \bibinfo{person}{Kun Gai}.} \bibinfo{year}{2023}\natexlab{}.
\newblock \showarticletitle{KuaiSim: {A} Comprehensive Simulator for Recommender Systems}. In \bibinfo{booktitle}{\emph{Advances in Neural Information Processing Systems 36: Annual Conference on Neural Information Processing Systems 2023, NeurIPS 2023, New Orleans, LA, USA, December 10 - 16, 2023}}, \bibfield{editor}{\bibinfo{person}{Alice Oh}, \bibinfo{person}{Tristan Naumann}, \bibinfo{person}{Amir Globerson}, \bibinfo{person}{Kate Saenko}, \bibinfo{person}{Moritz Hardt}, {and} \bibinfo{person}{Sergey Levine}} (Eds.).
\newblock
\urldef\tempurl%
\url{http://papers.nips.cc/paper\_files/paper/2023/hash/8c7f8f98f9a8f5650922dd4545254f28-Abstract-Datasets\_and\_Benchmarks.html}
\showURL{%
\tempurl}


\bibitem[Zheng et~al\mbox{.}(2024)]%
        {DBLP:conf/icde/ZhengHLCZCW24}
\bibfield{author}{\bibinfo{person}{Bowen Zheng}, \bibinfo{person}{Yupeng Hou}, \bibinfo{person}{Hongyu Lu}, \bibinfo{person}{Yu Chen}, \bibinfo{person}{Wayne~Xin Zhao}, \bibinfo{person}{Ming Chen}, {and} \bibinfo{person}{Ji{-}Rong Wen}.} \bibinfo{year}{2024}\natexlab{}.
\newblock \showarticletitle{Adapting Large Language Models by Integrating Collaborative Semantics for Recommendation}. In \bibinfo{booktitle}{\emph{40th {IEEE} International Conference on Data Engineering, {ICDE} 2024, Utrecht, The Netherlands, May 13-16, 2024}}. \bibinfo{publisher}{{IEEE}}, \bibinfo{pages}{1435--1448}.
\newblock
\href{https://doi.org/10.1109/ICDE60146.2024.00118}{doi:\nolinkurl{10.1109/ICDE60146.2024.00118}}


\bibitem[Zhou et~al\mbox{.}(2025a)]%
        {zhou2025onerec}
\bibfield{author}{\bibinfo{person}{Guorui Zhou}, \bibinfo{person}{Jiaxin Deng}, \bibinfo{person}{Jinghao Zhang}, \bibinfo{person}{Kuo Cai}, \bibinfo{person}{Lejian Ren}, \bibinfo{person}{Qiang Luo}, \bibinfo{person}{Qianqian Wang}, \bibinfo{person}{Qigen Hu}, \bibinfo{person}{Rui Huang}, \bibinfo{person}{Shiyao Wang}, {et~al\mbox{.}}} \bibinfo{year}{2025}\natexlab{a}.
\newblock \showarticletitle{OneRec Technical Report}.
\newblock \bibinfo{journal}{\emph{arXiv preprint arXiv:2506.13695}} (\bibinfo{year}{2025}).
\newblock


\bibitem[Zhou et~al\mbox{.}(2025b)]%
        {zhou2025onerecv2}
\bibfield{author}{\bibinfo{person}{Guorui Zhou}, \bibinfo{person}{Hengrui Hu}, \bibinfo{person}{Hongtao Cheng}, \bibinfo{person}{Huanjie Wang}, \bibinfo{person}{Jiaxin Deng}, \bibinfo{person}{Jinghao Zhang}, \bibinfo{person}{Kuo Cai}, \bibinfo{person}{Lejian Ren}, \bibinfo{person}{Lu Ren}, \bibinfo{person}{Liao Yu}, \bibinfo{person}{Pengfei Zheng}, \bibinfo{person}{Qiang Luo}, \bibinfo{person}{Qianqian Wang}, \bibinfo{person}{Qigen Hu}, \bibinfo{person}{Rui Huang}, \bibinfo{person}{Ruiming Tang}, \bibinfo{person}{Shiyao Wang}, \bibinfo{person}{Shujie Yang}, \bibinfo{person}{Tao Wu}, \bibinfo{person}{Wuchao Li}, \bibinfo{person}{Xinchen Luo}, \bibinfo{person}{Xingmei Wang}, \bibinfo{person}{Yi Su}, \bibinfo{person}{Yunfan Wu}, \bibinfo{person}{Zexuan Cheng}, \bibinfo{person}{Zhanyu Liu}, \bibinfo{person}{Zixing Zhang}, \bibinfo{person}{Bin Zhang}, \bibinfo{person}{Boxuan Wang}, \bibinfo{person}{Chaoyi Ma}, \bibinfo{person}{Chengru Song}, \bibinfo{person}{Chenhui Wang}, \bibinfo{person}{Chenglong Chu},
  \bibinfo{person}{Di Wang}, \bibinfo{person}{Dongxue Meng}, \bibinfo{person}{Dunju Zang}, \bibinfo{person}{Fan Yang}, \bibinfo{person}{Fangyu Zhang}, \bibinfo{person}{Feng Jiang}, \bibinfo{person}{Fuxing Zhang}, \bibinfo{person}{Gang Wang}, \bibinfo{person}{Guowang Zhang}, \bibinfo{person}{Han Li}, \bibinfo{person}{Honghui Bao}, \bibinfo{person}{Hongyang Cao}, \bibinfo{person}{Jiaming Huang}, \bibinfo{person}{Jiapeng Chen}, \bibinfo{person}{Jiaqiang Liu}, \bibinfo{person}{Jinghui Jia}, \bibinfo{person}{Kun Gai}, \bibinfo{person}{Lantao Hu}, \bibinfo{person}{Liang Zeng}, \bibinfo{person}{Qiang Wang}, \bibinfo{person}{Qidong Zhou}, \bibinfo{person}{Rongzhou Zhang}, \bibinfo{person}{Shengzhe Wang}, \bibinfo{person}{Shihui He}, \bibinfo{person}{Shuang Yang}, \bibinfo{person}{Siyang Mao}, \bibinfo{person}{Sui Huang}, \bibinfo{person}{Tiantian He}, \bibinfo{person}{Tingting Gao}, \bibinfo{person}{Wei Yuan}, \bibinfo{person}{Xiao Liang}, \bibinfo{person}{Xiaoxiao Xu}, \bibinfo{person}{Xugang Liu},
  \bibinfo{person}{Yan Wang}, \bibinfo{person}{Yang Zhou}, \bibinfo{person}{Yi Wang}, \bibinfo{person}{Yiwu Liu}, \bibinfo{person}{Yue Song}, \bibinfo{person}{Yufei Zhang}, \bibinfo{person}{Yunfeng Zhao}, \bibinfo{person}{Zhixin Ling}, {and} \bibinfo{person}{Ziming Li}.} \bibinfo{year}{2025}\natexlab{b}.
\newblock \showarticletitle{OneRec-V2 Technical Report}.
\newblock \bibinfo{journal}{\emph{CoRR}}  \bibinfo{volume}{abs/2508.20900} (\bibinfo{year}{2025}).
\newblock
\showeprint[arXiv]{2508.20900}
\href{https://doi.org/10.48550/ARXIV.2508.20900}{doi:\nolinkurl{10.48550/ARXIV.2508.20900}}


\end{thebibliography}

\end{document}